\documentclass{article}

\usepackage[final]{neurips_2020}
\usepackage[page]{appendix}

\usepackage[colorlinks=true,citecolor=black,linkcolor=blue]{hyperref}

\usepackage[utf8]{inputenc} 
\usepackage[T1]{fontenc}    
\usepackage{hyperref}       
\usepackage{url}            
\usepackage{booktabs}       
\usepackage{amsfonts}       
\usepackage{nicefrac}       
\usepackage{microtype}      

\usepackage{tikz}
\usepackage{thmtools,thm-restate}
\usepackage{subcaption}
\usepackage{amssymb}
\usepackage{algorithm}
\usepackage{algorithmic}
\usepackage{paralist}
\usepackage{caption}
\usepackage{comment}
\usepackage{flushend}

\newcommand{\todo}[1]{{\color{blue}todo: #1}}

\usepackage{amsmath,amsfonts,amsthm,amssymb}
\usepackage{bbm}
\usepackage{bm}

\newtheorem{defn}{Definition}
\newtheorem{lemma}{Lemma}

\newtheorem{prop}{Proposition}
\newtheorem{cor}{Corollary}
\newtheorem*{remark*}{Remark}

\newcommand{\tT}{\tilde{T}}

\newcommand{\cB}{\mathcal{B}}
\newcommand{\cC}{\mathcal{C}}
\newcommand{\cE}{\mathcal{E}}

\newcommand{\cI}{\mathcal{I}}

\newcommand{\cP}{\mathcal{P}}
\newcommand{\cR}{\mathcal{R}}

\newcommand{\cT}{\mathcal{T}}

\newcommand{\bbE}{\mathbb{E}}

\newcommand{\bbP}{\mathbb{P}}

\newcommand{\pa}{{\textrm{pa}}}

\newcommand{\adj}{{\textrm{adj}}}
\newcommand{\CC}{{\texttt{CC}}}

\newcommand{\rvD}{{\mathsf{D}}}
\newcommand{\DCT}{{\textrm{DCT}}}
\newcommand{\obs}{{\textrm{obs}}}
\newcommand{\res}{{\textrm{res}}}
\DeclareMathOperator{\Res}{Res}

\DeclareMathOperator{\skel}{skel}

\DeclareMathOperator{\Residuals}{\cR}
\DeclareMathOperator{\compRatio}{R}
\newcommand{\up}{{\textrm{up}}}

\newcommand{\cb}[1]{\left\lfloor \frac{#1}{2} \right\rfloor}

\usepackage{prettyref}
\newcommand{\rref}[2][]{\prettyref{#2}}
\newrefformat{model}{Model\,\ref{#1}}
\newrefformat{listing}{Listing\,\ref{#1}}
\newrefformat{alg}{Algorithm\,\ref{#1}}
\newrefformat{line}{line\,\ref{#1}}
\newrefformat{sec}{Section\,\ref{#1}}
\newrefformat{subsec}{Subsection\,\ref{#1}}
\newrefformat{section}{Section\,\ref{#1}}
\newrefformat{appendix}{Appendix\,\ref{#1}}
\newrefformat{app}{Appendix\,\ref{#1}}
\newrefformat{def}{Definition\,\ref{#1}}
\newrefformat{defn}{Definition\,\ref{#1}}
\newrefformat{thm}{Theorem\,\ref{#1}}
\newrefformat{ax}{\ref{#1}}
\newrefformat{prop}{Prop.\,\ref{#1}}
\newrefformat{lemma}{Lemma\,\ref{#1}}
\newrefformat{cor}{Corollary\,\ref{#1}}
\newrefformat{corollary}{Corollary\,\ref{#1}}
\newrefformat{ex}{Example\,\ref{#1}}
\newrefformat{tab}{Table\,\ref{#1}}
\newrefformat{fig}{Fig.\,\ref{#1}}
\newrefformat{eqn}{Equation~(\ref{#1})}
\newrefformat{problem}{Problem\,\ref{#1}}

\newtheorem*{prop*}{Prop.}
\newtheorem*{lemma*}{Lemma}
\newtheorem*{defn*}{Definition}

\newcommand{\Max}{{\textrm{max}}}

\newcommand{\leftarrowstar}{\kern1.5pt\hbox{$\leftarrow$}\kern-1.5pt\hbox{$\ast$}\kern1.5pt}
\newcommand{\rightarrowstar}{\kern1.5pt\hbox{$\ast$}\kern-1.5pt\hbox{$\rightarrow$}}
\newcommand{\tailstar}{\kern1.5pt\hbox{$-$}\kern-1.5pt \hbox{$\ast$}}
\newcommand{\startail}{\kern1.5pt\hbox{$\ast$}\kern-1.5pt\hbox{$-$}}
\newcommand{\staredge}{\kern1.5pt\hbox{$\ast$}\kern-1.5pt\hbox{$-$}\kern-1.5pt\hbox{$\ast$}\kern1.5pt}

\begin{document}

\title{Active Structure Learning of Causal DAGs via Directed Clique Trees}

\author{%
  Chandler Squires \\
  LIDS, MIT\\
  MIT-IBM Watson AI Lab\\
  \texttt{csquires@mit.edu} \\
   \And
  Sara Magliacane \\
  MIT-IBM Watson AI Lab \\
  IBM Research \\ 
  \texttt{sara.magliacane@gmail.com} \\
   \And
  Kristjan Greenewald \\
  MIT-IBM Watson AI Lab \\
  IBM Research \\ 
  \texttt{kristjan.h.greenewald@ibm.com} \\
   \And
   Dmitriy Katz \\
  MIT-IBM Watson AI Lab \\
  IBM Research \\ 
  \texttt{dkatzrog@us.ibm.com} 
  \And
    Murat Kocaoglu \\
  MIT-IBM Watson AI Lab \\
  IBM Research \\ 
  \texttt{murat@ibm.com} 
  \And
    Karthikeyan Shanmugam \\
  MIT-IBM Watson AI Lab \\
  IBM Research \\ 
  \texttt{karthikeyan.shanmugam2@ibm.com} 
}
\maketitle

\begin{abstract}
A growing body of work has begun to study intervention design for efficient structure learning of causal directed acyclic graphs (DAGs).
A typical setting is a \emph{causally sufficient} setting, i.e. a system with no latent confounders, selection bias, or feedback, when the essential graph of the observational equivalence class (EC) is given as an input and interventions are assumed to be noiseless.
Most existing works focus on \textit{worst-case} or \textit{average-case} lower bounds for the number of interventions required to orient a DAG. These worst-case lower bounds only establish that the largest clique in the essential graph \textit{could} make it difficult to learn the true DAG. 
In this work, we develop a \textit{universal} lower bound for single-node interventions that establishes that the largest clique is \textit{always} a fundamental impediment to structure learning.
Specifically, we present a decomposition of a DAG into independently orientable components through \emph{directed clique trees} and use it to prove that 
the number of single-node interventions necessary to orient any DAG in an EC is at least the sum of half the size of the largest cliques in each chain component of the essential graph.
Moreover, we present a two-phase intervention design algorithm that, under certain conditions on the chordal skeleton, matches the optimal number of interventions up to a multiplicative logarithmic factor in the number of maximal cliques. 
We show via synthetic experiments that our algorithm can scale to much larger graphs than most of the related work and achieves better worst-case performance than other scalable approaches.
\footnote{A code base to recreate these results can be found at \url{https://github.com/csquires/dct-policy}.}
\end{abstract}

\section{Introduction}\label{section:introduction}
Causal modeling is an important tool in medicine, biology and econometrics, allowing practitioners to predict the effect of actions on a system and the behavior of a system if its causal mechanisms change due to external factors \citep{Pearl:2009:CMR:1642718, Spirtes2000,PetJanSch17}. 
A commonly-used model is the directed acyclic graph (DAG), which is capable of modeling \emph{causally sufficient} systems, i.e. systems with no latent confounders, selection bias, or feedback. 
However, even in this favorable setup, a causal model cannot (in general) be fully identified from observational data alone; in these cases experimental (``interventional'') data is necessary to resolve ambiguities about causal relationships.

In many real-world applications, interventions may be time-consuming or expensive, e.g. randomized controlled trials or gene knockout experiments.
These settings crucially rely on \emph{intervention design}, i.e. finding a cost-optimal set of interventions that can fully identify a causal model. 
Recently, many methods have been developed for intervention design under different assumptions
\citep{he2008active,hyttinen2013experiment,shanmugam2015learning,kocaoglu2017cost,lindgren2018experimental}.

In this work we extend the Central Node algorithm of \citet{greenewald2019trees} to learn the structure of causal graphs in a \emph{causally sufficient} setting from interventions on single variables for both noiseless and noisy interventions. Noiseless interventions are able to deterministically orient a set of edges, while noisy interventions result in a posterior update over a set of compatible graphs.
We also focus only on interventions with a single target variable, i.e. \emph{single-node interventions}, but as opposed to \citep{greenewald2019trees} which focuses on limited types of graphs, we allow for general DAGs but only consider noiseless interventions.
%
In particular, we focus on \textit{adaptive} intervention design, also known as \textit{sequential} or \textit{active} \citep{he2008active}, where the result of each intervention is incorporated into the decision-making process for later interventions. 
This contrasts with \textit{passive} intervention design, for which all interventions are decided beforehand.

\textbf{Universal lower bound.} Our key contribution is to show that the problem of fully orienting a DAG with single-node interventions is equivalent to fully orienting special induced subgraphs of the DAG, called \emph{residuals} (\rref{thm:vis-characterization} below).
Given this decomposition, we prove a universal lower bound on the minimum number of single-node interventions necessary to fully orient \textit{any} DAG in a given Markov Equivalence Class (MEC), the set of graphs that fit the observational distribution. This lower bound is equal to the sum of half the size of the largest cliques in each chain component of the essential graph (\rref{thm:clique-lower-bound}).
This result has a surprising consequence: the largest clique is \textit{always} a fundamental impediment to structure learning. 
In comparison, prior work \citep{hauser2014two,shanmugam2015learning} established \textit{worst-case} lower bounds based on the maximum clique size, which only implied that the largest clique in each chain component of the essential graph \textit{could} make it difficult to learn the true DAG.

\textbf{Intervention policy.} We also propose a novel two-phase single-node intervention policy. 
The first phase, based on the Central Node algorithm, uses properties of \emph{directed clique trees} (\rref{def:dct}) to reduce the identification problem to identification within the (DAG dependent) residuals. 
The second phase then completes the orientations within each residual. 
We cover the condition of \emph{intersection-incomparability} for the chordal skeleton of a DAG
(\cite{kumar2002clique} introduce this condition in the context of graph theory)
. We show that under this condition, our policy uses at most $O(\log \cC_\Max)$ times as many interventions as are used by the (DAG dependent) optimal intervention set, where $\cC_\Max$ is the greatest number of maximal cliques in any chain component (\rref{thm:dct-policy-competitive-ratio}).

Finally, we evaluate our policy on general synthetic DAGs. 
We find that our intervention policy performs comparably to intervention policies in previous work, while being much more scalable than most policies and adapting more effectively to the difficulty of the underlying identification problem.

\section{Preliminaries}\label{section:preliminaries}
We briefly review our notation and terminology for graphs. 
A mixed graph $G$ is a tuple of vertices $V(G)$, directed edges $D(G)$, bidirected edges $B(G)$, and undirected edges $U(G)$. Directed, bidirected, and undirected edges between vertices $i$ and $j$ in $G$ are denoted $i \rightarrow_G j$, $i \leftrightarrow_G j$, and $i -_G j$, respectively. We use asterisks as wildcards for edge endpoints, e.g., $i \rightarrowstar_G j$ denotes either $i \rightarrow_G j$ or $i \leftrightarrow_G j$. 
A \textit{directed cycle} in a mixed graph is a sequence of edges $i \rightarrowstar_G \ldots \rightarrowstar_G~i$ with at least one directed edge.
A mixed graph is a \textit{chain graph} if it has no directed cycles and $B(G) = \emptyset$, and a chain graph is called a \textit{directed acyclic graph} (\textit{DAG}) if we also have $U(G) = \emptyset$. An \emph{undirected graph} is a mixed graph with $B(G) = \emptyset$ and $D(G) = \emptyset$.

\textbf{DAGs and ($\cI$-)Markov equivalence.} DAGs are used to represent causal models \citep{Pearl:2009:CMR:1642718}. Each vertex $i$ is associated with a random variable $X_i$. The \emph{skeleton} of graph $D$, $\skel(D)$, is the undirected graph with the same vertices and adjacencies as $D$.
A distribution $f$ is \emph{Markov} w.r.t. a DAG $D$ if it factors as $f(X) = \prod_{i \in V(D)} f(X_i \mid X_{\pa_D(i)})$. Two DAGs $D_1$ and $D_2$ are called \textit{Markov equivalent} if all positive distributions $f$ which are Markov to $D_1$ are also Markov to $D_2$ and vice versa. 
The set of DAGs that are Markov equivalent to $D$ is the \emph{Markov equivalence class} (MEC), denoted  as $[D]$. 
$[D]$ is represented by a chain graph called the \emph{essential graph} $\cE(D)$, which has the same skeleton as $D$, with directed edges $i \rightarrow_{\cE(D)} j$ if $i \rightarrow_{D'} j$ for all $D' \in [D]$, and undirected edges otherwise.
Given an intervention $I \subseteq V(D)$, the distributions $(f^\obs, f^I)$ are \emph{I-Markov} to $D$ if $f^\obs$ is Markov to $D$ and $f^I$ factors as 
\begin{eqnarray*}
    f^I(X) = \prod_{i \not\in I} f^\obs(X_i \mid X_{\pa_D(i)}) 
    \prod_{i \in I} f^I(X_i \mid X_{\pa_D(i)})
\end{eqnarray*}
where $\pa_D(i)$ represents the set of parents of vertex $i$ in the DAG $D$.
Given a list of interventions $\cI = [I_1, \ldots, I_M]$, the set of distributions $\{f^\obs, f^{I_1}, \ldots, f^{I_M}\}$ is \emph{$\cI$-Markov} to a DAG $D$ if $(f^\obs, f^{I_m})$ is $I_m$-Markov to $D$ for $\forall m = 1 \dots M$. 
The \emph{$\cI$-Markov equivalence class} of $D$ ($\cI$-MEC), denoted as $[D]_{\cI}$,   
can be represented by the \emph{$\cI$-essential graph} $\cE_\cI(D)$ with the same adjacencies as $D$ and $i \rightarrow_{\cE_\cI(D)} j$ if $i \rightarrow_{D'} j$ for all $D' \in [D]_\cI$.

The edges which are \textit{undirected} in the essential graph $\cE(D)$, but \textit{directed} in the $\cI$-essential graph $\cE_\cI(D)$, are the edges which are learned from performing the interventions in $\cI$. 
In the special case of a single-node intervention, the edges learned are all of those \textit{incident} to the intervened node, along with any edges learned via the set of logical constraints known as Meek rules \rref{app:meek}.



\textbf{Structure of essential graphs.} We now report a known result that proves that any intervention policy can split essential graphs in components that can be oriented independently.
The \textit{chain components} of a chain graph $G$, denoted $\CC(G)$, are the 
connected components of the graph after removing its directed edges. These chain components are then clearly undirected graphs.
An undirected graph is \textit{chordal} if every cycle of length greater than 3 has a \textit{chord}, i.e., an edge between two non-consecutive vertices in the cycle.
\begin{lemma}[\citet{hauser2014two}]\label{lemma:hauser}
Every $\cI$-essential graph is a chain graph with chordal chain components.
Orientations in one chain component do not affect orientations in other components.
\end{lemma}

\begin{defn}
    A DAG whose essential graph has a single chain component is called a \textit{moral DAG}.
\end{defn}

In many of the following results we will consider moral DAGs, since once we can orient moral DAGs we can easily generalize to general DAGs through these results. 

\textbf{Intervention Policies.} An \emph{intervention policy} $\pi$ is a (possibly randomized) map from ($\cI$-)essential graphs to interventions. An intervention policy is \textit{adaptive} if each intervention $I_m$ is decided based on information gained from previous interventions, and \textit{passive} if the whole set of interventions $\cI$ is decided prior to any interventions being performed. An intervention is \textit{noiseless} if the intervention set $\cI$ collapses the set of compatible graphs exactly to the $\cI$-MEC, while \textit{noisy} interventions simply induce a posterior update on the distribution over compatible graphs. Most policies assume that the MEC is known (e.g., it has been estimated from observational data) and interventions are noiseless; this is true of our policy too.
Moreover, we focus only on interventions on a single target variable, i.e. \emph{single-node interventions}.
We discuss previous work on intervention policies in \rref{section:related-work}.

\section{Universal lower-bound in the number of single-node interventions}\label{sec:dct}

In this section we prove a lower-bound on any possible single-node policy (\rref{thm:clique-lower-bound}) by decomposing the complete orientation of a DAG in terms of the complete orientation of smaller independent subgraphs, called \emph{residuals} (\rref{thm:vis-characterization}), defined on a novel graphical structure, \emph{directed clique trees} (DCTs). 
We provide all proofs in the Appendix.

\begin{figure*}[t]
    \centering
    \includegraphics[width=.9\textwidth]{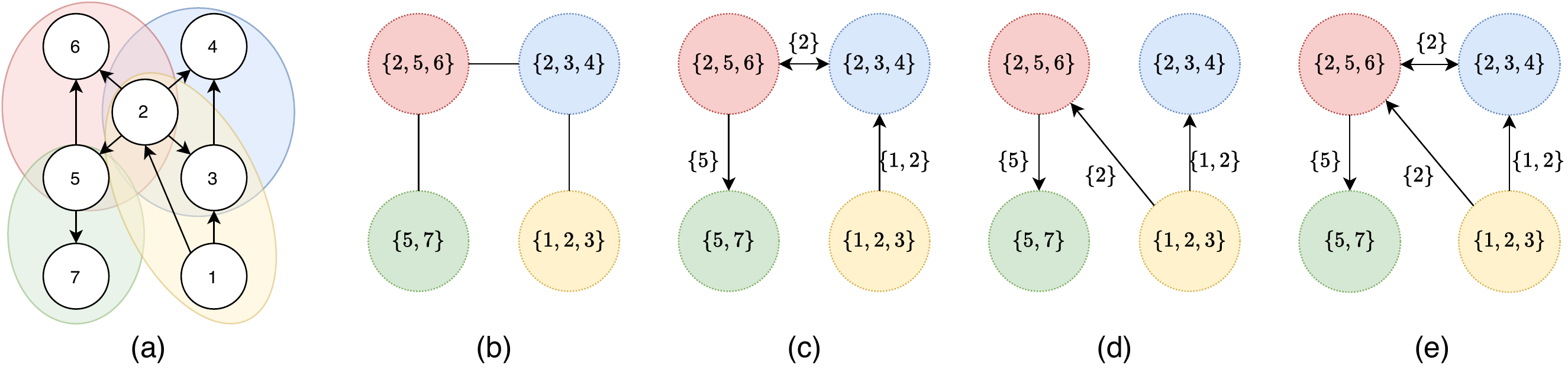}
    \caption{A moral DAG (a), one of its clique trees (b), its two DCTs (c-d) and the DCG (e).}
    \label{fig:dct-example}
\end{figure*}


First, we review the standard definitions of clique trees and clique graphs for undirected chordal graphs (see also \citep{galinier1995chordal}).
A \emph{clique} $C \subseteq V(G)$ is a subset of the nodes with an edge between each pair of nodes. A clique $C$ is \emph{maximal} if $C \cup \{v\}$ is not a clique for any $v \in V(G) \setminus C$. The set of maximal cliques of $G$ is denoted $\cC(G)$. The \textit{clique number} of $G$ is $\omega(G) = \max_{C \in \cC} |C|$. 
A \emph{clique tree} (aka a junction tree) $T_G$ of a chordal graph is a tree with vertices $\cC(G)$ that satisfies the \emph{induced subtree property}, i.e., for any $v \in V(G)$, the induced subgraph on the set of cliques containing $v$ is a tree. 
A chordal graph can have multiple clique trees, so we denote
the set of all clique trees of $G$ as $\cT(G)$. A \textit{clique graph} $\Gamma_G$ is the graph union of all clique trees, i.e. the undirected graph with $V(\Gamma_G) = \cC(G)$ and $U(\Gamma_G) = \cup_{T \in \cT(G)} U(T)$.  A useful characterization of the clique trees of $G$ are as the max-weight spanning trees of the weighted clique graph $W_G$ \citep{koller2009probabilistic}, which is a complete graph over vertices $\cC(G)$, with the edge $C_1 -_{W_G} C_2$ having weight $|C_1 \cap C_2|$.

Given a moral DAG $D$, we can trivially define its clique trees $\cT(D)$ as the clique trees of its skeleton $G = \skel(D)$, i.e. $\cT(G)$.
For example, in \rref{fig:dct-example} (a) we show a DAG, where we have chosen a color for each of the cliques, while in \rref{fig:dct-example} (b) we show one of its clique trees.
We now define a directed counterpart to clique trees based on the orientations in the underlying DAG:

\begin{defn}\label{def:dct}
    A \emph{directed clique tree} $T_D$ of a moral DAG $D$ has the same vertices and adjacencies as a clique tree $T_G$ of $G = \skel(D)$. For each ordered pair of adjacent cliques $C_1 \staredge C_2$ we orient the edge mark of $C_2$ as:
    \begin{compactitem}
    \item $C_1 \rightarrowstar C_2$, if $\forall v_{12} \in C_1 \cap C_2$ and $\forall v_2 \in C_2 \setminus C_1$, we have $v_{12} \rightarrow_D v_2$ in the DAG $D$;
    \item  $C_1 \startail C_2$ otherwise, i.e. if there exists at least one incoming edge from $C_2 \setminus C_1$ into $C_1 \cap C_2$,
    \end{compactitem}
     
    where we recall that $*$ denotes a \textit{wildcard} for an edge. Thus, the above conditions only decide the presence or absence of an arrowhead at $C_2$; the presence or absence of an arrowhead at $C_1$ is decided when considering the reversed order.
\end{defn}

\begin{figure}[!t]
\begin{minipage}[t]{0.48\textwidth}
\centering
    \includegraphics[width=\textwidth]{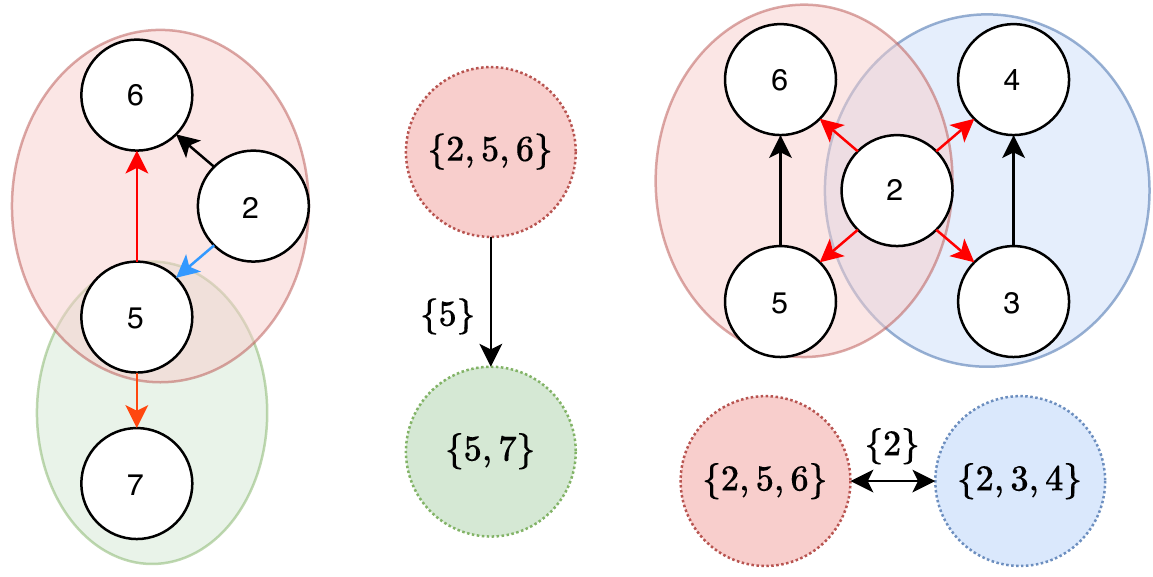}
    \caption{Examples of edge orientations.}
    \label{fig:edges-example}
\end{minipage} 
\begin{minipage}[t]{0.48\textwidth}
\centering
    \includegraphics[width=\textwidth]{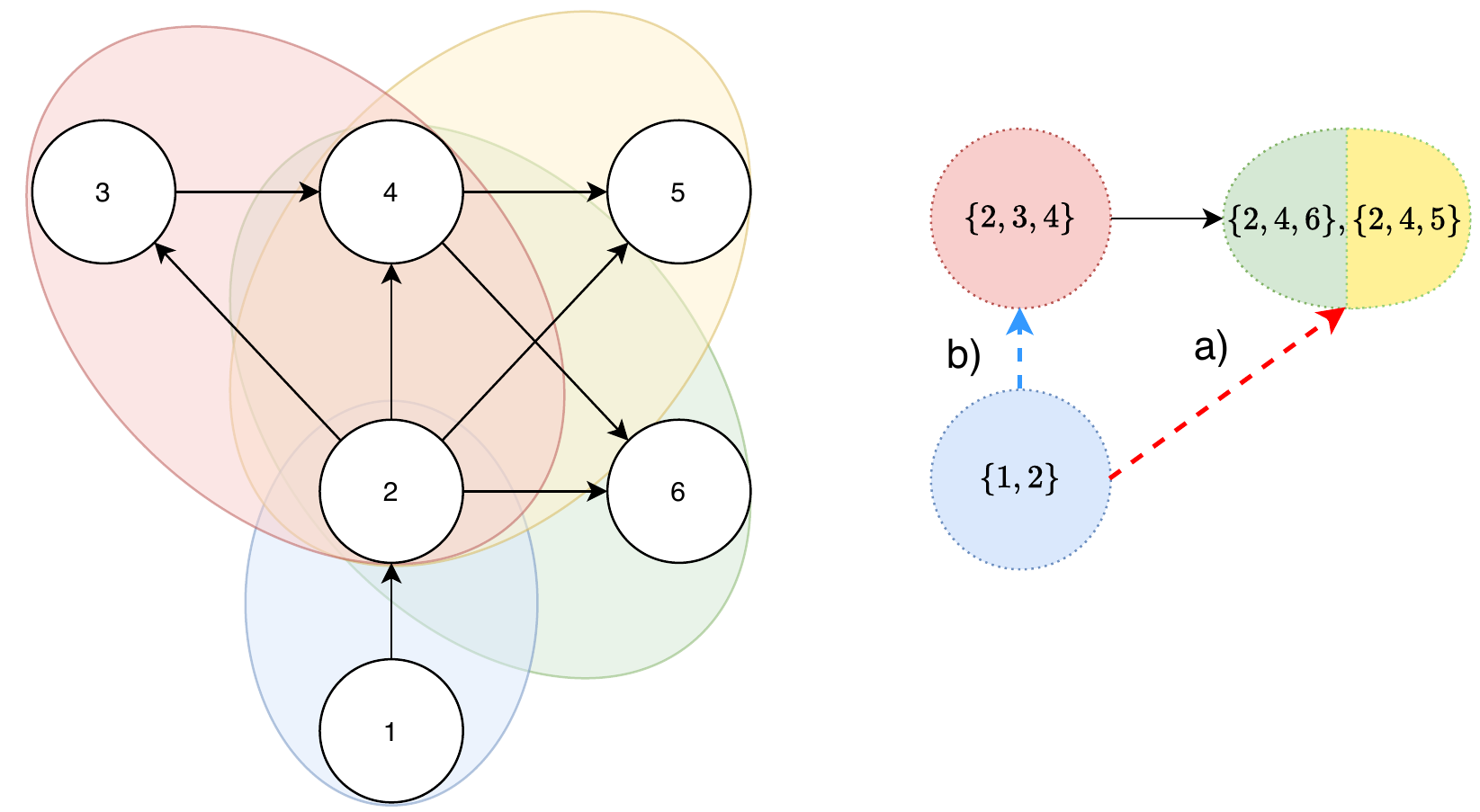}
    \caption{A DAG and its CDCTs with (using only edge a) and without arrow-meets (edge b).}
    \label{fig:conflicting-sources}
\end{minipage}
\end{figure}

A DAG can have multiple directed clique trees (DCTs), as shown in \rref{fig:dct-example} (c) and (d). In figures, we annotate edges with the intersection between cliques. \rref{fig:dct-example} (c) represents the directed clique tree corresponding to the standard clique tree in \rref{fig:dct-example} (b). In \rref{fig:edges-example} we show in detail the orientations for two of the directed clique edges following \rref{def:dct}, the red edges are outcoming from the clique intersection, while the blue edge is incoming in the intersection.
\rref{def:dct} also implies each edge that is shared between two different clique trees has a unique orientation (since it is based on the underlying DAG), so we can define the directed clique graph (DCG) $\Gamma_D$ of a moral DAG $D$ as the graph union of all directed clique trees of $D$. We show an example of a DCG in \rref{fig:dct-example}(e).
%
As can be seen in the examples in \rref{fig:dct-example}, DCTs can contain directed and bidirected edges, and, as we prove in Appendix \ref{app:proof-collider-implies-inclusion}, no undirected edgees. We define the bidirected components of a DCT as:
\begin{defn}\label{def:bidircomp}
The \emph{bidirected components} of $T_D$, $\cB(T_D)$, are the connected components of $T_D$ after removing directed edges.
\end{defn}

Another structure that can happen in a DCT is when two arrows meet at the same clique.
To avoid confusing associations with colliders in DAGs, we call these structures in DCTs \textit{arrow-meets}.
Arrow-meets will prove to be challenging for our algorithms, so we introduce \emph{intersection incomparability} and prove that in case it holds there can be no arrow-meets:
\begin{defn}\label{def:intersection-incomparable}
A pair of edges $C_1 -_{T_G} C_2$ and $C_2 -_{T_G} C_3$ are \emph{intersection comparable} if $C_1 \cap C_2  \subseteq C_2 \cap C_3$ or $C_1 \cap C_2  \supseteq C_2 \cap C_3$. Otherwise they are \emph{intersection incomparable}.
\end{defn}

For example, in \rref{fig:dct-example} (e), the edges 
$\{2, 5, 6\} {\leftrightarrow} \{2,3,4\}$ and 
$\{2,3,4\} {\leftarrow} \{1,2,3\}$ 
are intersection comparable, since $\{2\} \subset \{1,2\}$, while 
$\{2,5,6\} {\leftrightarrow} \{2,3,4\}$ 
and $\{2,5,6\} {\rightarrow} \{5, 7\}$
are intersection incomparable, since $\{2\} \not \subseteq \{5\}$ and $\{5\} \not \subseteq \{2\}$.

\begin{restatable}[]{prop}{arrowMeetsProp}\label{prop:collider-implies-inclusion}
Suppose $C_1 \rightarrowstar_{T_D} C_2$ and $C_2 \leftarrowstar_{T_D} C_3$ in $T_D$. Then these edges are intersection comparable.
Equivalently in the contrapositve, if $C_1 \rightarrowstar_{T_D} C_2$ and $C_2 \staredge_{T_D} C_4$ are intersection incomparable, 
we can immediately deduce that $C_2 \rightarrow_{T_D} C_4$.
\end{restatable}

Bidirected components do not have a clear ordering, so we contract them into single nodes in a contracted DCT (CDCT), and prove we can always construct a tree-like CDCT for any moral DAG:
\begin{defn}\label{def:contracted_dct}
The \emph{contracted directed clique tree (CDCT)} $\tT_D$ of a DCT $T_D$ is a graph on the vertex set $B_1, B_2 \dots B_K \in \cB(T_D)$ with $B_1 \rightarrow_{\tT_D} B_2$ if $C_1 \rightarrow_{T_D} C_2$ for any clique $C_1 \in B_1$ and $C_2 \in B_2$.
\end{defn}

\begin{restatable}[]{lemma}{uniqueSourceLemma}\label{lemma:construction-proof}
For any moral DAG $D$, one can always construct a CDCT with no arrow-meets.
\end{restatable}

In particular, one can adapt Kruskal's algorithm for finding a max-weight spanning tree to construct a DCT from the weighted clique graph and then contract it, as shown in detail in \rref{alg:dct-construction} in Appendix \ref{app:construction-proof}. In \rref{fig:conflicting-sources} we show an example of a CDCT with arrow-meets (represented by the black edge and the edge labelled ``a'') and its equivalent no arrow-meets version (represented by the black edge and the edge ``b'') . Since we can always construct a CDCT with no arrow-meets, we assume w.l.o.g. that the CDCT is a tree.
%
The CDCT allows us to define a decomposition of a moral DAG into independently orientable components. We call these components \emph{residuals}, since they
extend the notion of residuals in rooted, undirected clique trees \citep{vandenberghe2015chordal}. Formally:
%

\begin{figure}[!t]
\begin{minipage}[b]{0.57\textwidth}
    \centering
    \includegraphics[width=0.95\textwidth]{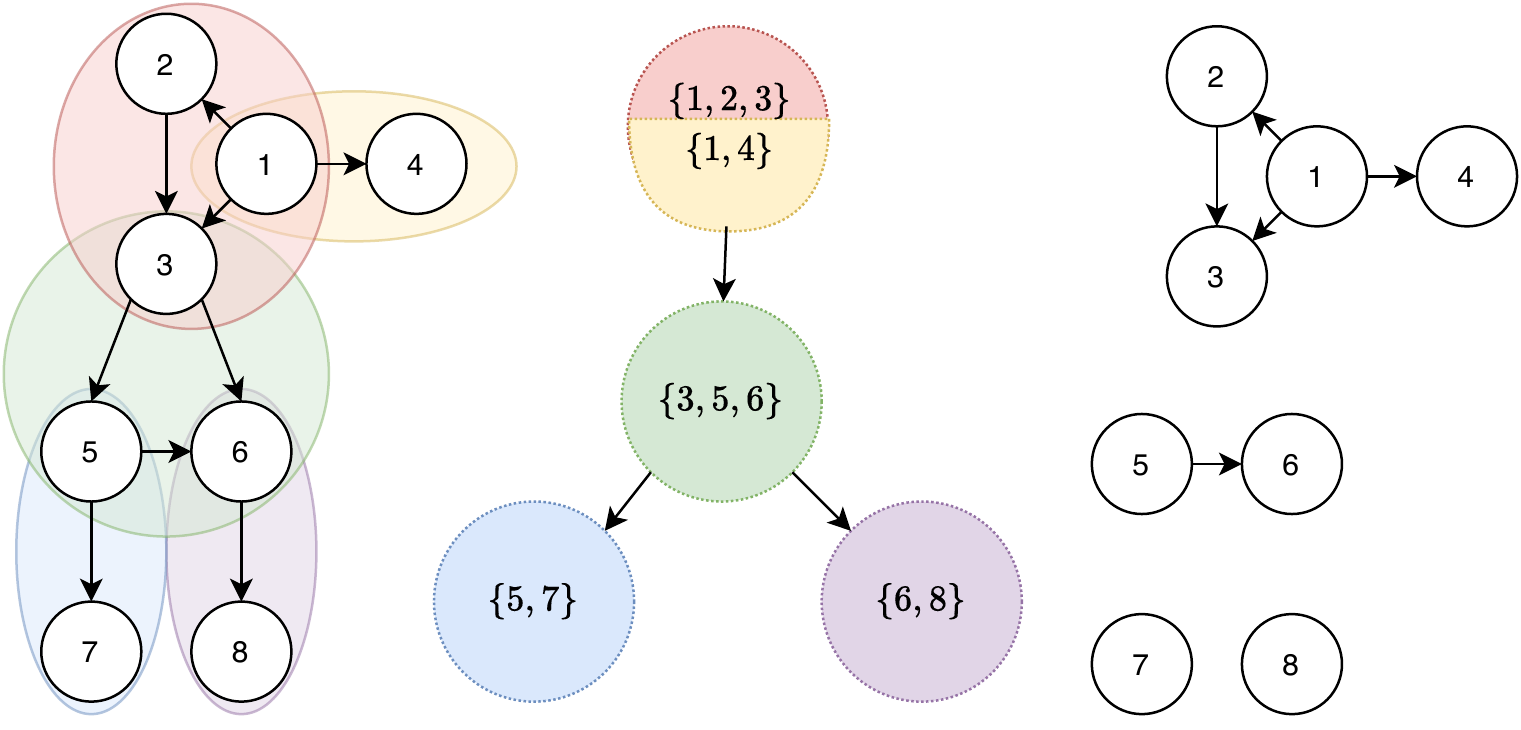}
    \caption{A DAG, its CDCT and its residuals. }
    \label{fig:residuals-example}
\end{minipage}
\begin{minipage}[b]{0.42\textwidth}
    \centering
    \includegraphics[width=0.9\textwidth]{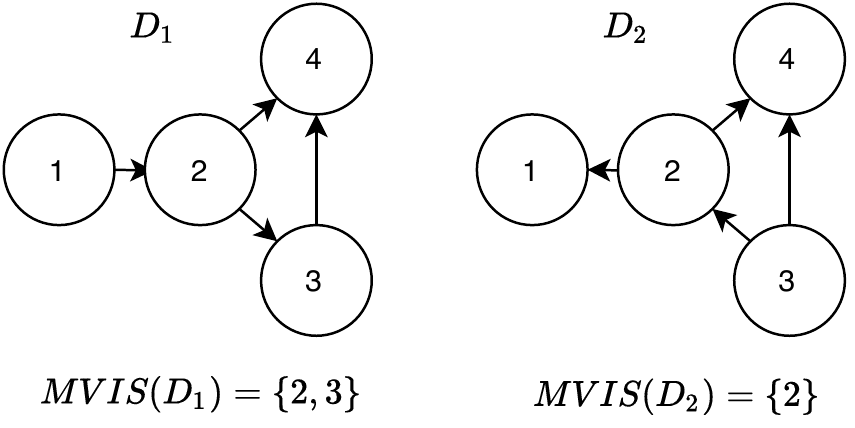}
    \caption{DAGs in the same MEC with $m(D_1) \neq m(D_2)$. }
    \label{fig:different_mvis}
\end{minipage}
\end{figure}

\begin{defn}\label{def:residual}
For a tree-like CDCT $\tT_D$ of a moral DAG $D$, the \emph{residual} of its node $B$
is defined as $\Res_{\tT_D}(B) = D|_{B \setminus P}$, where $P$ is its parent in $\tT_D$ (or if there is none, $P = \emptyset$) and $D|_{B \setminus P}$ is the induced subgraph of $D$ over the subset of $V(D)$ that are assigned to $B$ but not to $P$.
We denote the set of all residuals of $\tT_D$ by $\Residuals(\tT_D)$.
\end{defn}
Intuitively this describes the subgraphs in which we cut all edges that are captured in the CDCT, as shown in \rref{fig:residuals-example}. 
We now generalize our results from a moral DAG to a general DAG.
Surprisingly, we show that orienting all of the residuals for all chain components in the essential graph is both necessary and sufficient to \emph{completely orient any DAG}. 
We start by introducing a \emph{VIS}:
\begin{defn}\label{def:vis}
Given a general DAG $D$, a \emph{verifying intervention set ({VIS})} is a set of single-node interventions $\cI$ that fully orients the DAG starting from an essential graph, i.e. $\cE_\cI(D) = D$. A minimal VIS (MVIS) is a VIS of minimal size. We denote the size of the minimal VIS for $D$ as $m(D)$.
\end{defn}

For each DAG there are many possible VISes. A trivial VIS for any DAG is just the set of all of its nodes. In general, we are more interested in MVISes, which are also not necessarily unique for a DAG. For example, the DAG in \rref{fig:residuals-example} has four MVISes: $\{1,3,5\}$, $\{1,3,6\}$, $\{2,4,5\}$, and $\{2, 4, 6\}$.

We now show that finding a VIS for any DAG $D$ can be decomposed twice: first we can create a separate task of finding a VIS for each of the chain components $G$ of its essential graph $\cE(D)$, and then for each $G$ we can create a tree-like CDCT and find independently a VIS for each of its residuals:

\begin{restatable}[]{thm}{visDecompositionTheorem}\label{thm:vis-characterization}
    A single-node intervention set is a VIS for any general DAG $D$ iff it contains VISes for each residual $R \in \Residuals(\tT_G)$ for all chain components $G \in \CC(\cE(D))$ of its essential graph $\cE(D)$. 
\end{restatable}

An MVIS of $D$ will then contain only the MVISes of each residual of each chain component. An algorithm using this decomposition to compute an MVIS is given in Appendix \ref{app:brute-force-mvis}. 
In general, the size of an MVIS of $D$ cannot be calculated from just its essential graph, as shown by the two graphs in \rref{fig:different_mvis}. 
Instead, we propose a \emph{universal lower bound} that holds for all DAGs in the same MEC:



\begin{restatable}[]{thm}{lowerboundThm}\label{thm:clique-lower-bound}
Let $D$ be any DAG. Then $m(D) \geq \sum_{G \in \CC(\cE(D))} \cb{\omega(G)}$, where $\omega(G)$ is the size of the largest clique in each of the chain components $G$ of the essential graph $\cE(D)$.
\end{restatable}

We reiterate how this bound is different from previous work. 
For a fixed MEC $[D]$ with essential graph $\cE$, it is easy to construct $D^* \in [D]$ such that $m(D^*) \geq \sum_{G \in \CC(\cE(D))} \cb{\omega(G)}$ by picking the largest clique in each chain component to be the 
upstream-most clique. The bound in \rref{thm:clique-lower-bound} gives a much stronger result: \textit{any} choice of DAG in the MEC requires this many single-node interventions.

\section{A two-phase intervention policy based on DCTs}\label{section:policy}
While in the previous section we started from a known DAG $D$ to construct a CDCT and then proved an universal lower bound on $m(D)$, in this section we focus on intervention design to learn the orientations of an unknown DAG starting from its observational essential graph. 
\rref{thm:vis-characterization} proves that to orient a DAG $D$, we only need to orient the residuals for each of its essential graph chain components. 
The definition of residuals requires the knowledge of a tree-like CDCT for each component, which can be easily derived from the directed clique graph (DCG) (e.g. through \rref{alg:dct-construction} in Appendix \ref{app:construction-proof}).
So, we propose a two phase policy, in which the first phase uses interventions to identify the DCG of each chain components, while the second phase uses interventions to orient each of the residuals, as described in \rref{alg:dct-policy}. 
We now focus on describing the first phase of the algorithm and start by introducing two types of abstract, higher-level interventions.

\begin{defn}
    A \emph{clique-intervention} on a clique $C$ is a series of single-node interventions that suffices to learn the orientation of all edges in $\Gamma_D$ that are incident on $C$. An \emph{edge-intervention} on an edge $C_1 -_{T_G} C_2$ is a series of single-node interventions that suffices to learn the orientation of $C_1 -_{T_D} C_2$.
\end{defn}

A trivial clique-intervention is intervening on all of $C$, and a trivial edge-intervention is intervening on all of $C_1 \cap C_2$. 
The clique- and edge- interventions we use in practice are outlined in \rref{app:clique-edge-interventions}.

\begin{minipage} [t]{0.4\textwidth}
\begin{algorithm}[H]
\caption{\textsc{DCT Policy}}
\label{alg:dct-policy}
\begin{algorithmic}[1]
    \STATE \textbf{Input:} essential graph $\cE(D)$
    \FOR {component $G$ in $\CC(\cE(D))$}
        \STATE create clique graph $\Gamma_G$
        \STATE $\Gamma_D$ = \textsc{FindDCG}($\Gamma_G$)
        \STATE convert $\Gamma_D$  to a CDCT $\tT_D$
        \FOR {clique $C$ in $\tT_D$}
            \STATE Let $R = \Res_{\tT_D}(C)$
            \STATE Intervene on nodes in $V(R)$ until $R$ is fully oriented
    \ENDFOR
    \ENDFOR
    \STATE \textbf{return}  completely oriented $D$
\end{algorithmic}
\end{algorithm}
\end{minipage}
\begin{minipage} [t]{0.59\textwidth}
\begin{algorithm}[H]
\caption{\textsc{FindDCG}}
\label{alg:phase1}
\begin{algorithmic}[1]
    \STATE \textbf{Input:} clique graph $\Gamma_G$
    \STATE let $\Gamma_D = \Gamma_G$
    \WHILE{$\Gamma_D$ has undirected edges}
        \STATE let $T$ be a maximum-weight spanning tree of the undirected component of $\Gamma_D$
        \STATE let $C$ be a central node of $T$ 
        \STATE perform a clique-intervention on $C$
        \STATE let $P_\up(C) = $ \textsc{IdentifyUpstream}$(C)$
        \STATE let $S = V(B_T^{C: P_\up(C)})$
        \WHILE{$\Gamma_D$ has unoriented incident to $\cC \setminus S$}
            \STATE propagate edges in $\Gamma_D$
            \STATE perform an edge-intervention on an edge $C_1 -_{\Gamma_G} C_2$ with $C_1 \in \cC \setminus S$
        \ENDWHILE
    \ENDWHILE
    \STATE \textbf{return} $\Gamma_D$
\end{algorithmic}
\end{algorithm}
\end{minipage}

The first phase of our algorithm, described in \rref{alg:phase1}, is inspired by the Central Node algorithm  \citep{greenewald2019trees}. This algorithm operates over a tree, so we will have to use a spanning tree:
\begin{defn}\citep{greenewald2019trees}
Given a tree $T$ and a node $v \in V(T)$, we divide $T$ into \textit{branches} w.r.t. $v$. For a node $w$ adjacent to $v$, the branch $B_T^{(v:w)}$ is the connected component of $T - \{v\}$ that contains $w$. 
A \emph{central node} $c$ is a node for which $\forall w$ adjacent to $c: |B_T^{(c:w)}| \leq |\frac{V(T)}{2}|$.
\end{defn}

While our algorithm works for general graphs, it will help our intuition to first assume that $\Gamma_G$ is intersection-incomparable. In this case, there are no arrow-meets in $\Gamma_D$ by \rref{prop:collider-implies-inclusion}, nor in any of the directed clique trees. Thus, after each clique-intervention on a central node $C$, there will be only one parent clique upstream and the algorithm will orient at least half of the remaining unoriented edges by repeated application of \rref{prop:collider-implies-inclusion}.
For the intersection-comparable case, two steps can go wrong. First, after a clique-intervention on $C$, we may find that $C$ has multiple parents in $\Gamma_D$ (i.e. $C$ is at an arrows-meet). We can prove that even in this case, there is always a single ``upstream'' branch, identified via the \texttt{IdentifyUpstream} procedure, described in \rref{app:identify-upstream}, which performs edge-interventions on a subset of the parents.
A second step which may go wrong is in the propagation of orientations along the downstream branches, which halts when encountering intersection-incomparable edges. In this case, we simply kickstart further propagation by performing an edge-intervention.




The size of the problem is cut in half after each clique-intervention, so that we use at most $\sum_{G\in \CC(\cE(D))} \lceil \log_2(|\cC(G)|) \rceil$ clique-interventions, where $\cC(G)$ is the set of maximal cliques for $G$.  Furthermore, if $\Gamma_G$ is intersection-incomparable we use no edge-interventions (see \rref{lemma:phase1} in \rref{app:dct-policy-competitive-ratio}).
The second phase of the algorithm then orients the residuals and uses at most $\sum_{G \in \CC(\cE(D))} \sum_{C \in \cC(G)} |\Res_{\tT_G}(C)| - 1$ single-node interventions (see \rref{lemma:phase2} in \rref{app:dct-policy-competitive-ratio}).

\begin{restatable}{thm}{thmPolicyCompetitive}\label{thm:dct-policy-competitive-ratio}
    Assuming $\Gamma_G$ is intersection-incomparable, Algorithm \ref{alg:dct-policy} uses at most 
    $(3 \lceil \log_2 \cC_\Max \rceil + 2) m(D)$
    single-node interventions, where $\cC_\Max = \max_{G \in \CC(\cE(D))} |\cC(G)|$.
\end{restatable}

In the extreme case in which the essential graph is a tree, a single intervention on the root node can orient the tree, so $m(D) = 1$, and $|\cC| = |V(D)| - 1$, so \rref{thm:dct-policy-competitive-ratio} says that Algorithm \ref{alg:dct-policy} uses $O(\log(p))$ interventions, which is the scaling of the Bayes-optimal policy for the uniform prior as discussed in \citet{greenewald2019trees}.

\textbf{Remark on intersection-incomparability.}
Intersection-incomparable chordal graphs were introduced as ``uniquely representable chordal graphs” in \cite{kumar2002clique}. This class was shown to include familiar classes of graphs such as proper interval graphs. While the assumption of intersection-incomparability is necessary for our analysis of the DCT policy, the policy still performs well on intersection-comparable graphs as demonstrated in \rref{section:experiments}. This suggests that the restriction may be an artifact of our analysis, and the result of \rref{thm:dct-policy-competitive-ratio} may hold more generally.

\section{Experimental Results}\label{section:experiments}

\begin{figure*}[t!]
    \centering
    \begin{subfigure}[b]{0.4\textwidth}
         \centering
         \includegraphics[width=\textwidth]{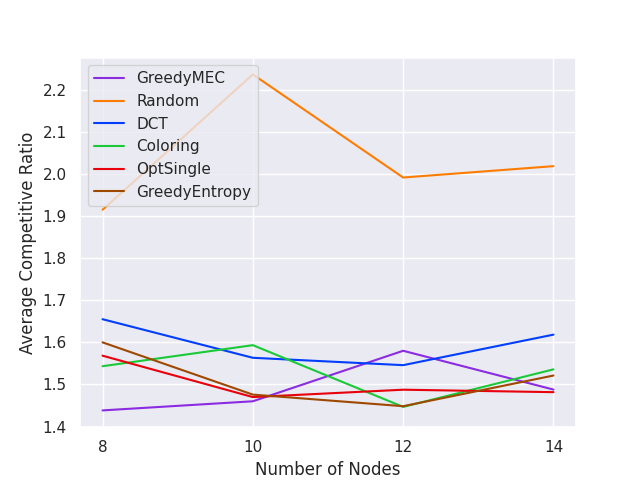}
         \caption{Average ic-Ratio (small graphs)}
         \label{fig:avg-regret-small}
     \end{subfigure}
     ~
     \begin{subfigure}[b]{0.4\textwidth}
         \centering
         \includegraphics[width=\textwidth]{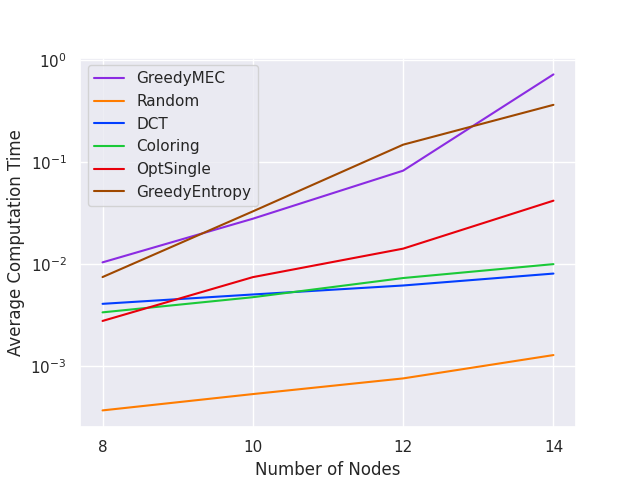}
         \caption{Average Time (small graphs)}
         \label{fig:time-small}
     \end{subfigure}
    
     \begin{subfigure}[b]{0.4\textwidth}
         \centering
         \includegraphics[width=\textwidth]{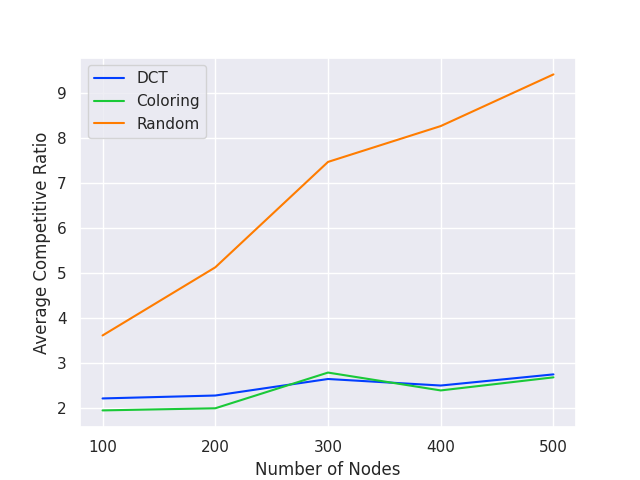}
         \caption{Average ic-Ratio (larger graphs)}
         \label{fig:avg-regret-large}
     \end{subfigure}
     ~
     \begin{subfigure}[b]{0.4\textwidth}
         \centering
         \includegraphics[width=\textwidth]{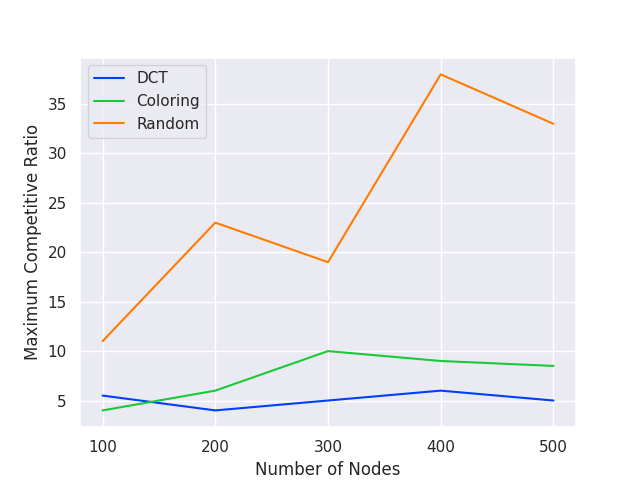}
         \caption{Maximum ic-Ratio (larger graphs)}
         \label{fig:maxregret-large}
     \end{subfigure}
     \caption{Comparison of intervention policies over 100 synthetic DAGs.}
\end{figure*}

We evaluate our policy on synthetic graphs of varying size.
To evaluate the performance of a policy on a specific DAG $D$, relative to $m(D)$, the size of its smallest VIS (MVIS), we adapt the notion of \textit{competitive ratio} from online algorithms \citep{borodin1992optimal,daniely2019competitive}. We use $\iota_D(\pi)$ to denote the expected size of the VIS found by policy $\pi$ for the DAG $D$, and define our evaluation metric as:

\begin{defn}
    The \emph{instance-wise competitive ratio} (\emph{ic-ratio}) of an intervention policy $\pi$ on $D$ is $\compRatio(\pi, D) = {\iota_{D'}(\pi) \over m(D') }$. The \emph{competitive ratio} on an MEC $[D]$ is $\compRatio(\pi) = max_{D' \in [D]}{\iota_{D'}(\pi) \over m(D') }$.
\end{defn}

The instance-wise competitive ratio of a policy on a DAG $D$ simply measures the number of interventions used by the policy \textit{relative} to the number of interventions used by the best policy \textit{for that DAG}, i.e., the policy which guesses that $D$ is the true DAG and uses exactly a MVIS of $D$ to verify this guess.
Thus, a \textit{lower} ic-ratio is better, and an ic-ratio of 1 is the best possible.
In order to compute the ic-ratio on $D$, we must compute $m(D)$, the size of a MVIS for $D$.
In our experiments, we use our DCT characterization of VIS's from \rref{thm:vis-characterization} to decompose the DAG into its residuals, each of whose MVIS's can be computed efficiently. We describe this procedure in \rref{app:brute-force-mvis}.

\textbf{Smaller graphs.} For our evaluation on smaller graphs, we generate random connected moral DAGs using the following procedure, which is a modification of Erd{\"o}s-R{\'e}nyi sampling that guarantees that the graph is connected. 
We first generate a random ordering $\sigma$ over vertices. 
Then, for the $n$-th node in the order, we set its indegree to be $X_n = \max(1, \textrm{Bin}(n-1, \rho))$, and sample $X_n$ parents uniformly from the nodes earlier in the ordering. 
Finally, we chordalize the graph by running the elimination algorithm \citep{koller2009probabilistic} with elimination ordering equal to the reverse of $\sigma$.

We compare the \texttt{OptSingle} policy \citep{hauser2014two}, the Minmax and Entropy strategy of \citet{he2008active}, called \texttt{MinmaxMEC} and \texttt{MinmaxEntropy}, respectively, and the coloring-based strategy of \citet{shanmugam2015learning}, called \texttt{Coloring}. 
We also introduce a baseline that picks randomly among non-dominated\footnote{A node is \emph{dominated} if all incident edges are directed, or if it has only a single incident edge to a neighbor with more than one incident undirected edges} nodes in the $\cI$-essential graph, called the \textit{non-dominated random} (\texttt{ND-Random}) strategy. As the name suggests, dominated nodes are easily proven to be non-optimal interventions, so \texttt{ND-Random} is a more fair baseline than simply picking randomly amongst nodes.

In \rref{fig:avg-regret-small} and \rref{fig:time-small}, we show the average ic-ratio and the average run-time for each of the algorithms. In terms of average ic-ratio, all algorithms aside from \texttt{ND-Random} perform comparably, using on average 1.4-1.7x more interventions than the smallest MVIS. However, the computation time grows quite quickly for \texttt{GreedyMEC}, \texttt{GreedyEntropy}, and \texttt{OptSingle}. This is because, when scoring a node as a potential intervention target, each of these algorithms iterates over all possible parent sets of the node. Moreover, the \texttt{GreedyMEC} and \texttt{GreedyEntropy} policies then compute the sizes of the resulting interventional MECs, which can grow superexponentially in the number of nodes \citep{gillispie2013enumerating}. In \rref{app:additional-experiments}, we show that in the same setting, \texttt{OptSingle} takes >10 seconds per graph for just 25 nodes, whereas \texttt{Coloring}, \texttt{DCT}, and \texttt{Random} remain under .1 seconds per graph.

\textbf{Larger graphs.} For our evaluation on large tree-like graphs, we create random moral DAGs of $n=100, \dots, 300$ nodes using the following procedure. We generate a complete directed 4-ary tree on $n$ nodes. Then, we sample an integer $R \sim U(2, 5)$ and add $R$ edges to the tree. Finally, we find a topological order of the graph by DFS and triangulate the graph using that order. This ensures that the graph retains a nearly tree-like structure, making $m(D)$ small compared to the overall number of nodes.
In \rref{fig:avg-regret-large} and \rref{fig:maxregret-large}, we show the average and maximum competitive ratio (computation time is given in \rref{app:additional-experiments}). For the average graph, our \texttt{DCT} policy and the \texttt{Coloring} policy use only 2-3 times as many interventions as the theoretical lower bound. Moreover, the worst competitive ratio experienced by the \texttt{DCT} algorithm is significantly smaller than the worst ratio experienced by the \texttt{Coloring} policy, which suggests that our policy is more adaptive to the underlying difficulty of the identification problem.

\section{Related Work}\label{section:related-work}
Intervention policies fall under two distinct, but related goals. 
The first is: given a fixed number of interventions, learn as much as possible about the underlying DAG. 
This goal is explored in \citet{ghassami2017optimal,ghassami2018budgeted} and \citet{hauser2014two}. 
The second goal, which is the one considered in this paper, is \textit{minimum-cost identification}: completely learn the underlying DAG using the least number of interventions.
We review previous work on policies operating under this objective.
As before, we use $\iota_D(\pi)$ to denote the expected size of the VIS found by policy $\pi$. 

We define $\Pi_K$ as the set of policies using interventions with at most $K$ target variables, i.e., $|I_m| \leq K$ for $I_m \in \cI$. 
We use $\Pi_\infty$ to represent policies allowing for interventions of unbounded size.
A policy $\pi$ is \emph{K-node minimax optimal} for an MEC $[D]$ if 
$\pi \in \arg\min_{\pi' \in \Pi_K} \max_{D' \in [D]} \iota_{D'}(\pi')$.
Informally, this is the policy $\pi$ that in the worst-case scenario (the DAG in the MEC that requires the most interventions under $\pi$) ends up requiring the least interventions. 
A policy is \emph{K-node Bayes-optimal} for an MEC $[D]$ and a prior $\bbP_{\rvD}$ supported only on the MEC $[D]$ if
$\pi \in \arg\min_{\pi' \in \Pi_K} \bbE_{\bbP_{\rvD}}[\iota_\rvD(\pi')]$

In the special cases of $K = 1$ and $K = \infty$, we replace $K$-node by \textit{single-node} and \textit{unbounded}, respectively. 
Much recent work explores intervention policies under a variety of objectives and constraints.
\citet{eberhardt2007causation} introduced passive, minimax-optimal intervention policies for single-node, $K$-node, and unbounded interventions in both the causally sufficient and causally insufficient case, when the MEC is not known. 
They also give a passive, unbounded intervention policy when the MEC \textit{is} known, and conjectures a minimax lower bound of $\lceil \frac{\omega(\cE(D))}{2} \rceil$ on $\iota_D(\pi)$ for such policies.
\citet{hauser2014two} prove this bound by developing a passive, unbounded minimax-optimal policy. 
\citet{shanmugam2015learning} develop a $K$-node minimax lower bound of $\frac{\omega(\cE(D))}{K} \log_{\frac{\omega(\cE(D))}{K}} \omega(\cE(D))$ based on separating systems. 
\citet{kocaoglu2017cost} develop a passive, unbounded minimax-optimal policy when interventions have distinct costs (where $\iota_D(\pi)$ is replaced by the total cost of all interventions.)
\citet{greenewald2019trees} develop an adaptive $K$-node intervention policy for noisy interventions which is within a small constant factor of the Bayes-optimal intervention policy, but the policy is limited to the case in which the chain components of the essential graph are trees.
It is important to note that all of these previous works give \textit{minimax} optimal policies, i.e. they focus on minimizing the interventions used in the \textit{worst} case over the MEC. In contrast, our result in \rref{thm:dct-policy-competitive-ratio} is \textit{competitive}, holding for \textit{every} DAG in the MEC, and shows that the largest clique is still a fundamental impediment to structure learning. However, the current result holds only in the single-node case, whereas previous work allows for larger interventions.

Finally, we note an interesting conceptual connection to \citet{ghassami2019counting}, which uses undirected clique trees as a tool for counting and sampling from MECs, suggesting that clique trees and their variants, such as DCTs, may be broadly useful for a variety of DAG-related tasks.

\section{Discussion}\label{section:discussion}
We presented a decomposition of a moral DAG into residuals, each of which must be oriented independently of one another. 
We use this decomposition to prove that for any DAG $D$ in a MEC with essential graph $\cE$, at least $\sum_{G \in \CC(\cE)} \cb{\omega(G)}$ interventions are necessary to orient $D$, where $\CC(\cE)$ denotes the chain components of $\cE$ and $\omega(G)$ denotes the clique number of $G$.
We introduced a novel two-phase intervention policy, which first uses a variant of the Central-Node algorithm to obtain orientations for the directed clique graph $\Gamma_D$, then orients within each residual. 
We showed that under certain conditions on the chain components of $\cE$, this intervention policy uses at most $(3 \log_2 \cC_\Max + 2)$ times as many interventions as the optimal intervention set. 
%
Finally, we showed on synthetic graphs that our intervention policy is more scalable than most existing policies, with comparable performance to the coloring-based policy of \cite{shanmugam2015learning} in terms of average ic-ratio and better performance in terms of worst-case ic-ratio. 

Preliminary results (\rref{app:additional-experiments}) suggest that the \DCT ~policy is more computationally efficient than the coloring-based policy on large, dense graphs, but is slightly worse in terms of performance. Further analysis of these results and possible improvements are left to future work.
Our results, especially the residual decomposition of the VIS, provide a foundation for further on intervention design in more general settings.

\flushcolsend
\clearpage

\newpage

\section*{Funding transparency statement}
Chandler Squires was supported by an NSF Graduate Research Fellowship and an MIT Presidential Fellowship and part of the work was performed during an internship at IBM Research. The work was supported by the MIT-IBM Watson AI Lab,

\section*{Broader impact statement}
Causality is an important concern in medicine, biology, econometrics and science in general \citep{Pearl:2009:CMR:1642718, Spirtes2000,PetJanSch17}. A causal understanding of the world is required to correctly
predict the effect of actions or external factors on a system, but also to develop fair algorithms. It is well-known that learning causal relations from observational data alone is not possible in general (except in special cases or under very strong assumptions); in these cases experimental (``interventional'') data is necessary to resolve ambiguities.

In many real-world applications, interventions may be time-consuming or expensive, e.g. randomized controlled trials to develop a new drug or gene knockout experiments. 
These settings crucially rely on \emph{experiment design}, or more precisely \emph{intervention design}, i.e. finding a cost-optimal set of interventions that can fully identify a causal model. 
The ultimate goal of intervention design is accelerating scientific discovery by decreasing its costs, both in terms of actual costs of performing the experiments and in terms of automation of new discoveries. 

Our work focuses on intervention design for learning causal DAGs, which have been notably employed as models in system biology, e.g. for gene regulatory networks \citep{friedman2000using} or for protein signalling networks \citep{sachs2005causal}. 
Protein signalling networks represent the way cells communicate with each other, and having reliable models of cell signalling is crucial to develop new treatments for many diseases, including cancer. Understanding how genes influence each other has also important healthcare applications, but is also crucial in other fields, e.g. agriculture or the food industry. Since even the genome of a simple organism as the common yeast contains 6275 genes, interventions like gene knockouts have to be carefully planned. 
Moreover, experimental design algorithms may prove to be a useful tool for driving down the time and cost of investigating the impact of cell type, drug exposure, and other factors on gene expression. These benefits suggest that there is a potential for experimental design algorithms such as ours to be a commonplace component of the future biological workflow. 

In particular, our work establishes a number of new theoretical tools and results that 1) may drive development of new experimental design algorithms, 2) allow practitioners to estimate, prior to beginning experimentation, how costly their task may be, 3) offer an intervention policy that is able to run on much larger graphs than most of the related work, and provides more efficient intervention schedules than the rest. 

Importantly, our work and in general intervention design algorithms have some limitations. In particular, as we have mentioned in the main paper, all these algorithms have relatively strong assumptions (e.g. no latent confounders or selection bias, infinite observational data, noiseless interventions, or in some case limitations on the graph structure \citep{greenewald2019trees}). If these assumptions are not satisfied in the data, or the practitioner does not realize their importance, the outcome of these algorithms could be misinterpreted or over-interpreted, leading to wasteful experiments or overconfident causal conclusions. Wrong causal conclusions may lead to potentially severe unintended side effects or unintended perpetuation of bias in algorithms.

Even in case of correct causal conclusions, the actualized impact of experimental design depends on the experiments in which it is used. Potential positive uses cases include decreasing the cost of drug development, in turn leading to better and cheaper medicine for consumers.

\newpage

\bibliography{bib.bib}
\flushcolsend

\clearpage
\section*{Supplementary material for: Active Structure Learning of Causal DAGs via Directed Clique Trees}
\appendix

\section{Meek Rules}\label{app:meek}
In this section, we recall the \textit{Meek rules} \citep{meek} for propagating orientations in DAGs. Of the standard four Meek rules, two of them only apply when the DAG contains v-structures. Since all DAGs that we need to consider do not have v-structures, we include only the first two rules here.

\begin{prop}[Meek Rules under no v-structures]
\quad 
\begin{enumerate}
    \item \textbf{No colliders:} If $a \rightarrow_G b -_G c$ and $a$ is not adjacent to $c$, then $b \rightarrow_G c$. \label{meekrule:no-collider}
    \item \textbf{Acyclicity:} If $a \rightarrow_G b \rightarrow_G c$ and $a$ is adjacent to $c$, then $a \rightarrow_G c$. \label{meekrule:no-cycle}
\end{enumerate}
\label{prop:meekRules}
\end{prop}

\section{The running intersection property}\label{app:running-intersection}

A useful and well-known property of clique trees, used throughout proofs in the remainder of the appendix, is the following:
\begin{prop*}[Running intersection property]\label{prop:running-intersection}
    Let $\gamma = \langle C_1, \ldots, C_K \rangle$ be the path between $C_1$ and $C_K$ in the clique tree $T_G$. Then $C_1 \cap C_K \subseteq C_k$ for all $C_k \in \gamma$. 
\end{prop*}

We refer the interested reader to \citet{maathuis2018handbook}.


\section{Proof of Proposition \ref{prop:collider-implies-inclusion}}\label{app:proof-collider-implies-inclusion}
This proposition describes the connection between arrow-meets and intersection comparability.
In order to prove this proposition, we begin by establishing the following propositions:
\begin{prop}
\label{prop:no-edge-adjacent}
Suppose $C_1$ and $C_2$ are adjacent in $T_G$. Then for all $v_1 \in C_1 \setminus C_2$, $v_2 \in C_2 \setminus C_1$, $v_1$ and $v_2$ are not adjacent in $G$.
\end{prop}
\begin{proof}
We prove the contrapositive. Suppose $v_1 \in C_1 \setminus C_2$ and $v_2 \in C_2 \setminus C_1$ are adjacent. Then $C_3' = (C_1 \cap C_2) \cup \{v_1, v_2\}$ is a clique and belongs to some maximal clique $C_3$. For the induced subtree property to hold, $C_3$ must lie between $C_1$ and $C_2$, i.e., $C_1$ and $C_2$ are not adjacent. 
\end{proof}

\begin{restatable}[]{prop}{radiatingIntersectionProp}\label{prop:radiating-intersection}
Let $D$ be a moral DAG, there are no undirected edges in any of its directed clique trees $T_D$, and therefore neither in its directed clique graph $\Gamma_D$.
\end{restatable}

\begin{proof}
(By contradiction). Suppose $v_1 \to_D v_{12}$ for $v_1 \in C_1 \setminus C_2$ and $v_{12} \in C_1 \cap C_2$. Suppose $v_2 \to_D v_{12}'$ for $v_2 \in C_2 \setminus C_1$, and $v_{12}' \in C_1 \cap C_2$. By the assumption that $D$ does not have v-structures and by \rref{prop:no-edge-adjacent}, $v_{12} \neq v_{12}'$.
Similarly, since $v_{12} \rightarrow_D v_2$  (otherwise there would be a v-structure with $v_1 \rightarrow_D v_{12}$) and $v_{12}' \rightarrow_D v_1$ (otherwise there would be a collider with $v_2 \rightarrow_D v_{12}'$). 
However, this induces a cycle $v_1 \rightarrow_D v_{12} \rightarrow_D v_2 \rightarrow_D v_{12}' \rightarrow_D v_1$.
\end{proof}

Now we can finally prove the final proposition:

\arrowMeetsProp*
\begin{proof}
We prove the contrapositive. If $C_1 \cap C_2 \not\subseteq C_2 \cap C_3$ and $C_1 \cap C_2 \not\supseteq C_2 \cap C_3$, then there exist nodes $v_{12} \in (C_1 \cap C_2) \setminus C_3$ and $v_{23} \in (C_2 \cap C_3) \setminus C_1$.
Since $v_{12}$ and $v_{23}$ are both in the same clique $C_2$ they are adjacent in the underlying DAG $D$, i.e. $v_{12} -_{D} v_{23}$.
Moreover since $C_1 \rightarrowstar_{T_D} C_2$ by 
the definition of a directed clique graph,
this edge is oriented as $v_{12} \rightarrow_D v_{23}$. Then by \rref{prop:radiating-intersection}, $C_2 \rightarrow_{T_D} C_3$.
\end{proof}

\section{Proof of \rref{lemma:construction-proof}}\label{app:construction-proof}

\uniqueSourceLemma*
\begin{proof}
To construct a CDCT with no arrow-meets, our approach is to first construct the DCT in a special way, so that after contraction, there are no arrow-meets. In particular, we need a DCT such that each bidirected component has at most one incoming edge. A DCT in which this does not hold is said to have \emph{conflicting sources}, formally:

\begin{defn}\label{def:noconflictsource}
A directed clique tree $T_D$ has two \emph{conflicting sources} $C_0$ and $C_{K+1}$, if $C_0 \rightarrow_{T_D} C_1$ and $C_{K} \leftarrow_{T_D} C_{K+1}$, and $C_1$ and $C_{K}$ are part of the same bidirected component $B \in \cB(T_D)$, i.e. $C_1, C_K \in B$, possibly with $C_1 = C_K$.
\end{defn}
An example of a clique tree with conflicting sources is given in \rref{fig:conflicting-sources2}.
The first DCT has conflicting sources $\{1,2\}$ and $\{2,3,4\}$, while the second DCT does not have conflicting sources.

\begin{figure}
\centering
    \includegraphics[width=\textwidth]{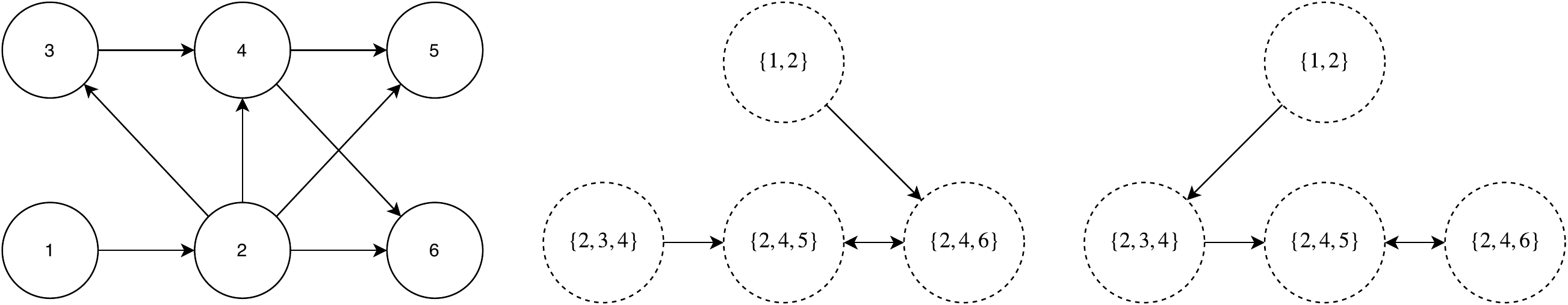}
    \caption{A DAG, its DCT with a conflicting source, and its DCG without a conflicting source.}
    \label{fig:conflicting-sources2}
\end{figure}

We will now show that \rref{alg:dct-construction} constructs a DCT with no conflicting sources. This is sufficient to prove \rref{lemma:construction-proof}, since after contraction, the resulting CDCT will have no arrow-meets.

First, \rref{alg:dct-construction} constructs a weighted clique graph $W_G$, which is a complete graph over vertices $\cC(G)$, with the edge $C_1 -_{W_G} C_2$ having weight $|C_1 \cap C_2|$. We will show that at each iteration $i$, there are no conflicting sources in $T_D$. This is clearly true for $i = 0$ since $T_D$ has no edges to begin.

At a given iteration $i$, suppose that the candidate edge $e = C_1 \rightarrowstar C_2$ is a maximum-weight edge that does not create a cycle, i.e. $e \in E$, but that it will induce conflicting sources. That is, the current $T_D$ already contains $C_2 \leftarrowstar C_3 \leftarrowstar \ldots \leftarrowstar C_{K-1} \leftarrow C_K$, where we choose $C_K$ that has no parents. Note that we can do this by following any directed/bidirected edges upstream (away from $C_2$), which must terminate since $T_D$ is a tree and thus does not have cycles.

By \rref{prop:collider-implies-inclusion}, $C_1 \cap C_2 \lesseqgtr C_2 \cap C_3$. In this case, $C_1 \cap C_2 \subseteq C_2 \cap C_3$, since $C_2 \leftarrowstar C_3$ was already picked as an edge and thus cannot have less weight (in other words, it cannot have a smaller intersection) than $C_1 \rightarrowstar C_2$. Furthermore, since $C_1 - C_2 - C_3$ is a valid subgraph of the clique tree, we must have $C_1 \cap C_3 \subseteq C_2$ by the running intersection property of clique trees (see \rref{app:running-intersection}). Combined with $C_1 \cap C_2 \subseteq C_2 \cap C_3$, we have $C_1 \cap C_3 = C_1 \cap C_2$. 
This means that $C_1 - C_3$ is also a valid edge in the weighted clique graph and it has the same weight ($C_1 \cap C_3$) as the $C_1 - C_2$ edge ($C_1 \cap C_2$). Moreover since $C_1 \rightarrowstar C_2$ then this edge will also preserve the same orientations $C_1 \rightarrowstar C_3$.
Thus, $C_1 \rightarrowstar C_3$ is another candidate maximum-weight edge that does not create a cycle. 
We may continue this argument, replacing $C_2$ by $C_k$, to show that $C_1 \rightarrowstar C_{K}$ is a maximum weight edge that does not create a cycle. Since $C_K$ has no parents, there are still no conflicting sources after adding $C_1 \rightarrowstar C_K$. Since we always pick a maximum-weight edge that does not create a cycle, this algorithm creates a maximum-weight spanning tree of $W_G$ \citep{koller2009probabilistic}, which is guaranteed to be a clique tree of $G$ \citet{koller2009probabilistic}.
\end{proof}

\begin{algorithm}[t]
\begin{algorithmic}[1]
    \STATE \textbf{Input:} DAG $D$
    \STATE let $W_G$ be the weighted clique graph of $G = \skel(D)$
    \STATE let $T_D$ be the empty graph over $V(W_G)$ 
    \FOR {$i = 1,\ldots,|V(W_G)|-1$}
        \STATE let $E$ be the set of maximum-weight edges of $W_G$ that do not create a cycle when added to $T_D$
        \STATE select $e \in E$ s.t. there are no conflicting sources\label{line:pick_edge}
        \STATE add $e$ to $T_D$
    \ENDFOR
    \STATE Contract the bidirected components of $T_D$ and create the CDCT $\tT_D$
    \STATE \textbf{Return} $\tT_D$
\end{algorithmic}
\caption{\textsc{Construct\_DCT}}
\label{alg:dct-construction}
\end{algorithm}


\section{Proof of Theorem \ref{thm:vis-characterization}}\label{app:proof-vis-characterization}

We restate the theorem here:
\visDecompositionTheorem*

In order to prove the following theorem we start by introducing a few useful concepts and results.

\begin{figure*}[t]
    \centering
    \includegraphics[width=.7\textwidth]{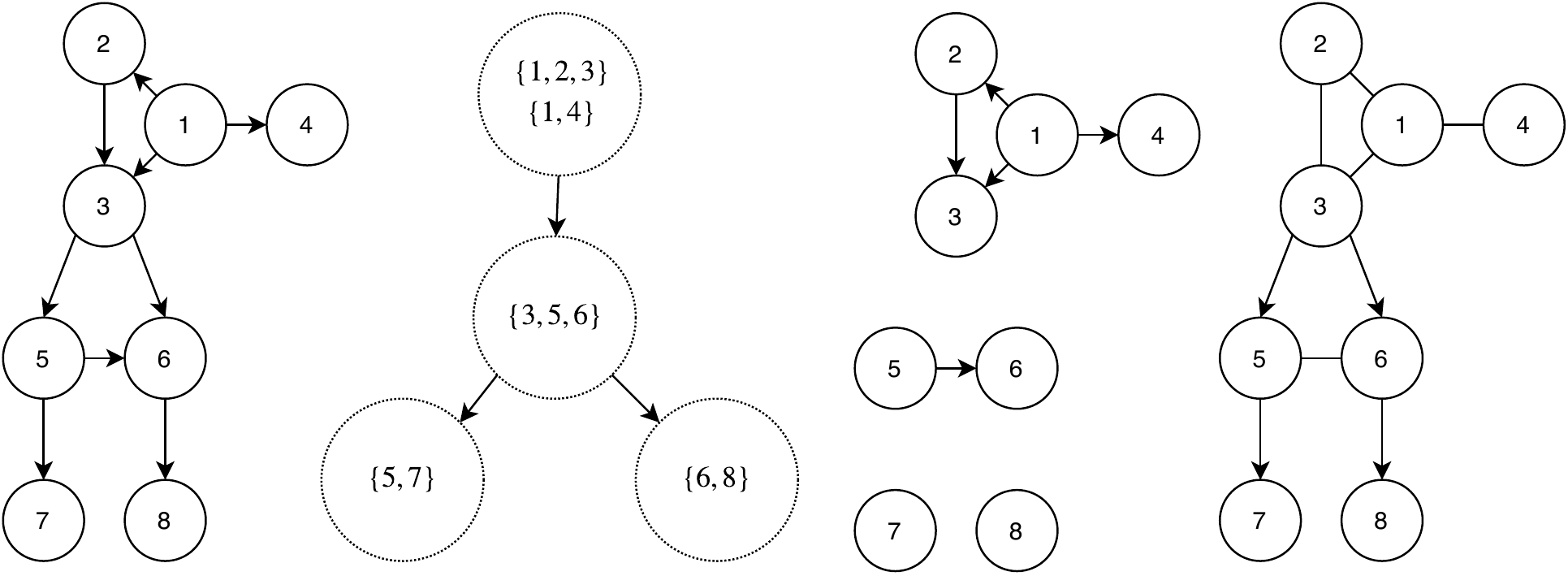}
    \caption{A DAG, its contracted directed clique tree, its residuals, and its residual essential graph.}
    \label{fig:residuals-example2}
\end{figure*}

\subsection{Residual essential graphs} 
The residuals decompose the DAG into parts which must be separately oriented. Intuitively, after adding orientations \textit{between} all pairs of residuals, the inside of one residual is cut off from the insides of other residuals. The following definition and lemmas formalize this intuition.

\begin{defn}\label{def:residual-essential-graph}
    The \emph{residual essential graph} $\cE_\res(D)$ of $D$ has the same skeleton as $D$, with $v_1 \rightarrow_{\cE_\res(D)} v_2$ iff $v_1 \rightarrow_D v_2$ and $v_1$ and $v_2$ are in different residuals of $\tT_D$.
\end{defn}

The following lemma establishes that after finding the orientations of edges in the DCT, the only remaining unoriented edges are in the residuals.

\begin{restatable}{lemma}{dctToReg}\label{lemma:dct-implies-res}
The oriented edges of $\cE_\res(D)$ can be inferred directly from the oriented edges of $T_D$.
\end{restatable}

\begin{proof}
In order to prove this theorem, we first introduce an alternative characterization of the residual essential graph defined only in terms of the orientations in the contracted DCT and prove its equivalence to \rref{def:residual-essential-graph}. Let $\cE'_\res(D)$ have the same skeleton as $D$, with $i \rightarrow_{\cE'_\res(D)} \!\!j$ if and only if $j \in \Res_{\tT_D}(B)$ and $i \in P$, for some $B \in \cB(T_D)$ and its unique parent $P$.

Suppose $v_1 \rightarrow_D v_2$ for $v_1 \in R_1$ and $v_2 \in R_2$, with $R_1, R_2 \in \cR(\tT_D)$ and $R_1 \neq R_2$.
Let $R_1 = \Res_{\tT_D}(B_1)$ and $R_2 = \Res_{\tT_D}(B_2)$ for $B_1, B_2 \in \cB(\tT_D)$.
There must be at least one clique $C_1 \in B_1$ that contains $v_1$, and likewise one clique $C_2 \in B_2$ that contains $v_2$.
Since $v_1$ and $v_2$ are adjacent, by the induced subtree property there must be some maximal clique on the path between $C_1$ and $C_2$ which contains $v_1$ and $v_2$.
Let $C_{12}$ be the clique on this path containing $v_1$ and $v_2$ that is closest to $C_1$.
Then, the next closest clique to $C_1$ must not contain $v_2$, so we will call this clique $C_{1 \setminus 2}$.
Since $v_1 \rightarrow_D v_2$, we know that $C_{1 \setminus 2} \rightarrow_{T_D} C_{12}$, hence $C_{1 \setminus 2}$ and $C_{12}$ are in different bidirected components, and thus $v_1 \rightarrow_{\cE_\res(D)} v_2$.
\end{proof}

\begin{restatable}{lemma}{regComplete}\label{lemma:residual-essential-graph-complete}
    The $\cE_\res(D)$ is complete under Meek's rules \citep{meek}.
\end{restatable}
\begin{proof}
Since Meek rules are sound and complete rules for orienting PDAGs \citep{meek}, and in our setting only two of the Meek rules apply (see \rref{prop:meekRules} in \rref{app:meek}), it suffices to show that neither applies for residual essential graphs.

First, suppose $i \rightarrow_{\cE_\res(D)} j$ and $j \rightarrow_{\cE_\res(D)} k$.
We must show that if $i$ and $k$ are adjacent, then $i \rightarrow_{\cE_\res(D)} k$, i.e. the acyclicity Meek rule does not need to be invoked.

We use the alternative characterization of $\cE_\res(D)$ from the proof of \rref{lemma:dct-implies-res}, which establishes that $i \rightarrow_\cE j$ iff. $j \in \Res_{\cT_D}(B)$ and $i \in P$ for some $B \in \cB(T_D)$ and its unique parent $P$.

Since $j \rightarrow_{\cE_\res(D)} k$, there must exist some component $B_{jk} \in \cB(T_D)$ containing $j$ and $k$ whose parent component $B_{j \setminus k}$ contains $j$ but not $k$, i.e. $B_{j \setminus k} \rightarrow_{\tT_D} B_{jk}$.
Likewise, there must be a component $B_{ij}$ containing $i$ and $j$ whose parent component $B_{i \setminus j}$ contains $i$ but not $j$, i.e. $B_{i \setminus j} \rightarrow_{\tT_D} B_{ij}$.
Moreover, since there is a clique on $\{i, j, k\}$, there must be at least one component $B_{ijk}$ containing $i$, $j$ and $k$.

We will prove that $B_{jk}$ and $B_{j \setminus k}$ both contain $i$, which implies $i \rightarrow_{\tT_D} k$.

Let $\gamma$ be the path in $\tT_D$ between $B_{i \setminus j}$ and $B_{jk}$.
This path must contain the edge $B_{j \setminus k} \rightarrow B_{jk}$, since $B_{i \setminus j}$ is upstream of $B_{jk}$, and $\cT_D$ is a tree.
By the induced subtree property on $k$, no component on the path other than $B_{jk}$ can contain $k$.
Now consider the path between $B_{ijk}$ and $B_{i \setminus j}$.
By the induced subtree property on $k$, this path must pass through $B_{jk}$.
Finally, by the induced subtree property on $i$, $B_{jk}$ and $B_{j \setminus k}$ must both contain $i$.

Now, we prove that also the first Meek rule is not invoked. Suppose $i \rightarrow_{\cE_\res(D)} j$, and $j$ is adjacent to $k$. We must show that if $i$ is not adjacent to $k$, then $j \rightarrow_{\cE_\res(D)} k$.

Since $\{i, j, k\}$ do not form a clique, there must be distinct components containing $i \rightarrow j$ and $j \rightarrow k$. Let $B_{ij}$ and $B_{jk}$ denote the closest such components in $\tT_D$, which are uniquely defined since $\tT_D$ is a tree. Since $i$ is upstream of $k$, $B_{ij}$ must be upstream of $B_{jk}$. Let $P := \pa_{\tT_D}(B_{jk})$, we know $j \in P$ since it is on the path between $B_{ij}$ and $B_{jk}$ (it is possible that $P = B_{ij}$). Since we picked $B_{jk}$ to be the closest component to $B_{ij}$ containing $\{j, k\}$, we must have $k \not\in P$, so indeed $j \rightarrow_G k$.
\end{proof}

For an example of the residual essential graph, see \rref{fig:residuals-example2}. \rref{lemma:residual-essential-graph-complete} implies that the residuals must be oriented separately, since the orientations in one do not impact the orientations in others.

\subsection{Proof for a moral DAG}
We then prove the result for a moral DAG $D$:

\begin{restatable}[VIS Decomposition]{lemma}{visDecompositionLemma}\label{lemma:vis-characterization}
    An intervention set is a VIS for a moral DAG $D$ iff it contains VISes for each residual of $\tT_D$. 
    This implies that finding a VIS for $D$ can be decomposed in several smaller tasks, in which we find a VIS for each of the residuals in $\Residuals(\tT_D)$.
\end{restatable}
\begin{proof}

~\\
\textbf{VISes of residuals are necessary.} We first prove that any VIS $\cI$ of $D$ must contain VISes for each residual of $D$. Consider the residual essential graph $\cE_\res(D)$ of $D$. We show that if we intervene on a node $c_1$ in the residual $R_1 = \Res_{\tT_D}(B_1)$ of some $B_1 \in \cB(\tT_D)$, then the only new orientations are between nodes in $R_1$, or in other words, each residual needs to be oriented independently.

By \rref{def:residual-essential-graph}, all edges between nodes in different residuals are already oriented in $\cE_\res(D)$.
A new orientation between nodes in $R_1$ will not have any impact for the nodes in the other residuals, which we can show by proving that Meek rules described in \rref{prop:meekRules} would not apply outside of the residual.
In particular, Meek Rule \ref{meekrule:no-collider} does not apply at all, since $b$ and $c$ must be in the same residual since the edge is undirected, but then $a$ is adjacent to $c$ since it's a clique.
Likewise, $a -_{\cE_\res(d)} c$, then $a$ and $b$ are in the same residual, so Meek Rule \ref{meekrule:no-cycle} only orients edges with both endpoints in the same residual.

\textbf{VISes of residuals are sufficient.} Now, we show that if $\cI$ contains VISes for each residual of $D$, then it is a VIS for $D$, i.e. that orienting the residuals will orient the whole graph by applying recursively Meek rules.
We will accomplish this by inductively showing that all edges in each bidirected component are oriented.
Let $\gamma = \langle B_1, \ldots, B_n \rangle$ be a path from the root of $\tT_D$ to a leaf of $\tT_D$.
As our base case, all edges in $B_1$ are oriented, since $B_1 = \Res_{\tT_D}(B_1)$.
Now, as our induction hypothesis, suppose that all edges in $B_{i-1}$ are oriented. 

The edges between nodes in $B_i$ are partitioned into three categories: edges with both endpoints also in $B_{i-1}$, edges with both endpoints in $\Res_{\tT_D}(B_i)$, and edges with one endpoint in $B_{i-1}$ and one endpoint in $\Res_{\tT_D}(B_i)$. 
The first category of edges are directed by the induction hypothesis, and the second category of edges are directed by the assumption that $\cI$ contains VISes for each residual.
It remains to show that all edges in the third category are oriented.
Each of these edges has one endpoint in some $C_{i-1} \in B_{i-1}$ and one endpoint in some $C_i$ in $B_i$, so we can fix some $C_{i-1}$ and $C_i$ and argue that all edges from $C_{i-1} \cap C_i$ to $C_i \setminus C_{i-1}$ are oriented.

Since $C_{i-1} \rightarrow_{R_D} C_i$, there exists some $c_{i-1} \in C_{i-1} \setminus C_i$ and $c' \in C_i \cap C_{i-1}$ such that $c_{i-1} \rightarrow_D c'$.
By \rref{prop:no-edge-adjacent}, $c_{i-1}$ is not adjacent to any $c_i \in C_i \setminus C_{i-1}$, so Meek Rule \ref{meekrule:no-collider} ensures that $c' \rightarrow_D c_i$ is oriented.
For any other node $c'' \in C_{i-1} \cap C_i$, either $c' \rightarrow_D c''$, in which case Meek Rule \ref{meekrule:no-cycle} ensures that $c_{i-1} \rightarrow_D c''$ and the same argument applies, or $c'' \rightarrow_D c'$, in which case Meek Rule \ref{meekrule:no-cycle} ensures that $c'' \rightarrow_D c_i$.
\end{proof}

\subsection{Proof for a general DAG}
We can now easily prove the theorem for any DAG $D$:
\visDecompositionTheorem*
\begin{proof}
    By the previous result (\rref{lemma:vis-characterization}) and \rref{lemma:hauser} from \citep{hauser2014two}.
\end{proof}

\section{Algorithm for finding an MVIS}\label{app:brute-force-mvis}
An algorithm using the decomposition into residuals to compute a minimal verifying intervention set (MVIS) is described in Algorithms \ref{alg:dct-verification} and \ref{alg:brute-force-mvis}.
Compared to running \rref{alg:brute-force-mvis} on any moral DAG, using \rref{alg:dct-verification} ensures that we only have to enumerate over subsets of the nodes in each residual, which in general require far fewer interventions.
Moreover, the residual of any component containing a single clique is itself a clique, which have easily characterized MVISes, and \rref{alg:brute-force-mvis} efficiently computes.

\begin{algorithm}[t]
\caption{\textsc{Find\_MVIS\_DCT}}
\label{alg:dct-verification}
\begin{algorithmic}[1]
    \STATE \textbf{Input:} Moral DAG $D$
    \STATE let $\tT_D$ be the contracted directed clique tree of $D$
    \STATE let $S = \emptyset$
    \FOR {component $B$ of $T_D$}
        \STATE let $R = \Res_{\tT_D}(B)$
        \STATE let $S' = $ \textsc{Find\_MVIS\_Enumeration}$(G[R])$
        \STATE let $S = S \cup S'$
    \ENDFOR
    \STATE \textbf{Return} $S$
\end{algorithmic}
\end{algorithm}

\begin{algorithm}[t]
\caption{\textsc{Find\_MVIS\_Enumeration}}
\label{alg:brute-force-mvis}
\begin{algorithmic}[1]
    \STATE \textbf{Input:} DAG $D$
    \IF{$D$ is a clique}
        \STATE Let $\pi$ be a topological ordering of $D$
        \STATE Let $S$ include even-indexed element of $\pi$
        \STATE \textbf{Return $S$}
    \ENDIF
    \FOR {$s=1, \ldots, |V(D)|$}
        \FOR {$S \subseteq V(D)$ with $|S| = s$}
            \IF {$S$ fully orients $D$}
                \STATE \textbf{Return} $S$
            \ENDIF
        \ENDFOR
    \ENDFOR
\end{algorithmic}
\end{algorithm}


\section{Proof of Theorem \ref{thm:clique-lower-bound}}\label{app:clique-lower-bound}
First, we prove the following proposition:

\begin{prop}\label{prop:bidirected-component-clique-bound}
Let $D$ be a moral DAG, $\cE = \cE(D)$ and let $\tT_D$ contain a single bidirected component. Then $m(D) \geq \cb{\omega(\cE)}$.
\end{prop}
\begin{proof}
    Let $C_1 \in \arg\max_{C \in \cC(\cE)} |C|$. By the running intersection property (see \rref{app:running-intersection}), for any clique $C_2$, $C_1 \cap C_2 \subseteq C_2 \cap C_\adj$ for $C_\adj$ adjacent to $C_2$ in $T_D$. Since $C_\adj \leftrightarrow_{T_D} C_2$, we have $v_{12} \rightarrow_D v_{2\setminus 1}$ for all $v_{12} \in C_1 \cap C_2$ and $v_{2\setminus 1} \in C_2 \setminus C_1$, i.e. there is no node in $D$ outside of $C_1$ that points into $C_1$. Thus, since the Meek rules only propagate downward, intervening on any nodes outside of $C_1$ does not orient any edges within $C_1$. 
    Finally, since $C_1$ is a clique, each consecutive pair of nodes in the topological order of $C_1$ must have at least one of the nodes intervened in order to establish the orientation of the edge between them.
    This requires at least $\left \lfloor\frac{|C_1|}{2}\right \rfloor$ interventions, achieved by intervening on the even-numbered nodes in the topological ordering.
\end{proof}

Now we can prove the following result for a moral DAG $D$:

\begin{restatable}{lemma}{lowerboundLemma}\label{lemma:lowerbound}

Let $D$ be a moral DAG and let $G = \skel(D)$. Then $m(D) \geq \cb{\omega(G)}$, where $\omega(G)$ is the size of the largest clique in $G$.
    
\end{restatable}

Consider a path $\gamma$ from the source of $\tT_D$ to the bidirected component containing the largest clique, i.e., $\gamma = \langle B_1, \ldots, B_Z \rangle$.
For each component, pick $C_i^* \in \arg\max_{C \in B_i} |C|$.
Also, let $R_i = \Res_{\tT_D}(B_i)$.
We will prove by induction that $\sum_{i=1}^z m(D[R_i]) \geq \max_{i=1}^z \cb{|C_i^*|}$ for any $z = 1, \ldots, Z$.
As a base case, it is true for $z = 1$, since $R_1 = B_1$ and by \rref{prop:bidirected-component-clique-bound}.

Suppose the lower bound holds for $z-1$. If $C_z^*$ is not the unique maximizer of $\left\lfloor\frac{|C_z^*|}{2}\right\rfloor$ over $i = 1, \ldots, z$, the lower bound already holds.
Thus, we consider only the case where $B_z$ is the unique maximizer.

Let $S_z = C_z^* \cap B_{z-1}$. By the running intersection property (see \rref{app:running-intersection}), $S_z$ is contained in the clique $C_\adj$ in $B_{z-1}$ which is adjacent to $C_z^*$ in $T_D$. Since $C_\adj$ is distinct from $C_z^*$, $|C_\adj^*| \geq |S_z| + 1$, and by the induction hypothesis we have that 
\begin{align*}
\sum_{i=1}^{z-1} m(D[R_i]) 
&\geq \max_{i=1,\ldots,z-1} \cb{|C_i^*|}
\\
&\geq \cb{ |C^*_{z-1}|}
\\
&\geq \cb{ |C_\adj| }
\\
&\geq \cb{|S_z|+1}
\end{align*}

Finally, applying \rref{prop:bidirected-component-clique-bound},
\begin{align*}
    \cb{|S_z+1|} + m(D[R_z]) 
    &\geq \cb{|S_z|+1} + \cb{|C_z^* \cap R_z|}
    \\
    &\geq \cb{|C_z^*|}
\end{align*}

where the last equality holds since $|S_z| + |C_z^* \cap R_z| = |C_z^*|$ and by the property of the floor function that $\cb{a+1} + \cb{b} \geq \cb{a+b}$, which can be easily checked.

Finally we can prove the theorem:

\lowerboundThm*
\begin{proof}
    By \rref{lemma:lowerbound} and Lemma 1 in \citet{hauser2014two}.
\end{proof}

\section{Clique and Edge Interventions}\label{app:clique-edge-interventions}
We present the procedures that we use for clique- and edge-interventions in \rref{alg:clique-intervention} and \rref{alg:edge-intervention}, respectively.

\begin{algorithm}[t]
\caption{\textsc{CliqueIntervention}}
\label{alg:clique-intervention}
\begin{algorithmic}[1]
    \STATE \textbf{Input:} Clique $C$
    \WHILE {$C -_{\Gamma_D} C'$ unoriented for some $C'$}
        \IF {$\exists v$ non-dominated in $C$}
            \STATE Pick $v \in C$ at random among non-dominated nodes.
        \ELSE
            \STATE Pick $v \in C$ at random.
        \ENDIF
        \STATE Intervene on $v$.
    \ENDWHILE
    \STATE \textbf{Output:} $P_\up(C)$
\end{algorithmic}
\end{algorithm}

\begin{algorithm}[t]
\caption{\textsc{EdgeIntervention}}
\label{alg:edge-intervention}
\begin{algorithmic}[1]
    \STATE \textbf{Input:} Adjacent cliques $C$, $C'$
    \WHILE{$C -_{\Gamma_D} C'$ unoriented}
        \STATE Pick $v \in C \cap C'$ at random.
        \STATE Intervene on $v$.
    \ENDWHILE
    \STATE \textbf{Output:} $P_\up(C)$
\end{algorithmic}
\end{algorithm}

\section{Identify-Upstream Algorithm}\label{app:identify-upstream}
Given the clique graph, a simple algorithm to identify the upstream branch consists of performing an edge-intervention on each pair of parents of $C$ to discover which is the most upstream.
However, if the number of parents of $C$ is large, this may consist of many interventions.
The following lemma establishes that the only parents which are candidates for being the most upstream are those whose intersection with $C$ is the smallest:

\begin{algorithm}[t]
\caption{\textsc{IdentifyUpstream}}
\label{alg:identify-upstream}
\begin{algorithmic}[1]
    \STATE \textbf{Input:} Clique $C$
    \FOR {$P_1, P_2 \in \cP_{\Gamma_D}(C)$}
        \STATE perform an edge-intervention on $P_1 -_{\Gamma_D} P_2$
    \ENDFOR
    \STATE \textbf{Output:} $P_\up(C)$
\end{algorithmic}
\end{algorithm}

\begin{restatable}{prop}{minimalParentProp}\label{prop:minimal-parent-upstream}
Let $P_\up(C) \in \pa_{\Gamma_D}(C)$ be the parent of $C$ which is upstream of all other parents. Then $P_\up(C) \in \cP_{\Gamma_D}(C)$, where $\cP_{\Gamma_D}(C)$ is the set of parents of $C$ in $\Gamma_D$ with the smallest intersection size, i.e., $P \in \cP_{\Gamma_D}(C)$ if and only if $P \rightarrow_{\Gamma_D} C$ and $|P \cap C| \leq |P' \cap C|$ for all $P' \in \pa_{\Gamma_D}(C)$.
\end{restatable}

\begin{proof}

We begin by citing a useful result on the relationship between clique trees and clique graphs when the clique contains an intersection-comparable edge:
\begin{lemma}[\citet{galinier1995chordal}]\label{lemma:inclusion-implies-bypass}
    If $C_1 -_{T_G} C_2 -_{T_G} C_3$ and $C_1 \cap C_2 \subseteq C_2 \cap C_3$, then $C_1 -_{\Gamma_G} C_3$.
\end{lemma}

\begin{cor}\label{cor:collider-implies-same-label}
    If $C_1 -_{T_G} C_2 -_{T_G} C_3$ and $C_1 \cap C_2 \subseteq C_2 \cap C_3$, then $C_1 \cap C_3 = C_1 \cap C_2$.
\end{cor}

\begin{proof}
    By the running intersection property of clique trees (see \rref{app:running-intersection}), $C_1 \cap C_3 \subseteq C_2$. Combined with $C_1 \cap C_2 \subseteq C_2 \cap C_3$ and simple set logic, the result is obtained.
\end{proof}

Every parent of $C$ is adjacent in $\Gamma_D$ to every other parent of $C$ by  \rref{prop:collider-implies-inclusion} and \rref{lemma:inclusion-implies-bypass}, and since every edge has at least one arrowhead, there can be at most one parent of $C$ that does not have an incident arrowhead.

Now we show that this parent must be in $\cP_{\Gamma_D}(C)$. \rref{cor:collider-implies-same-label} implies that for any triangle in $\Gamma_G$, two of the edge labels (corresponding to intersections of their endpoints) must be equal. If $P \in \cP_{\Gamma_D}(C)$ and $P' \in \pa_{T_D}(C) \setminus \cP_{\Gamma_D}(C)$, then the labels of $P \rightarrow_{\Gamma_D} C$ and $P' \rightarrow_{\Gamma_D} C$ are of different size and thus cannot match. Therefore, the label of $P \cap P' = P \cap C$. Finally, since we already know $P \rightarrow_{\Gamma_D} C$, it must also be the case that $P \rightarrow_{\Gamma_D} P'$.
\end{proof}




\section{Proof of \rref{thm:dct-policy-competitive-ratio}}\label{app:dct-policy-competitive-ratio}
We start by proving bounds for each of the two phases:
\begin{restatable}{lemma}{phaseOneLemma}\label{lemma:phase1}
    Algorithm \ref{alg:phase1} uses at most $\lceil \log_2 |\cC| \rceil$ clique-interventions. Moreover, assuming $T_G$ is intersection-incomparable, Algorithm \ref{alg:phase1} uses no edge-interventions.
\end{restatable}

\begin{proof}
    Since $T_G$ is intersection-incomparable, after a clique-intervention on $C$, orientations propagate in all but at most one branch of $T_G$ out of $C$. By the definition of a central node, the one possible remaining branch has at most half of the nodes from the previous time step, so the number of edges in $T_G$ reduces by at least half after each clique-intervention. Thus, there can be at most $\lceil \log_2 |\cC| \rceil$ clique-interventions.
 \end{proof}
 
For ease of notation, we will overload the symbol $\CC$ for the chain components of a chain graph $G$ to take a DAG as an argument, and return the subgraphs corresponding to the chain components of its essential graph. Formally, $\CC(D) = \{ D[V(G)] \mid G \in \CC(\cE(D)) \}$.
 
\begin{restatable}{lemma}{phaseTwoLemma}\label{lemma:phase2}
    The second phase of Algorithm \ref{alg:dct-policy} (line 6-8) uses at most $\sum_{C \in \cC(D')} |\Res_{\tT_{D'}}(C)| - 1$ single-node interventions for the moral DAG $D' \in \CC(D)$.
\end{restatable}
 
 \begin{proof}
     \citet{eberhardt2006n} show that $n-1$ single-node interventions suffice to determine the orientations of all edges between $n$ nodes. We sum this value over all residuals.
 \end{proof}

\thmPolicyCompetitive*

\begin{proof}
    Consider a moral DAG $D' \in \CC(D)$. We will show that \rref{alg:dct-policy} uses at most $(3 \lceil \log_2 |\cC(\cE(D))| \rceil + 2) m(D')$ single-node interventions. The result then follows since $m(D) = \sum_{D' \in \CC(D)} m(D')$, the total number of interventions used by \rref{alg:dct-policy} is the sum over the number interventions used for each chain component, and $\cC_\Max \geq |\cC(\cE(D))|$ for all $D'$.
    
    Assume that for each clique-intervention in Algorithm \ref{alg:phase1}, we intervene on every node in the clique. Then, the number of single-node interventions used by each clique intervention is upper-bounded by $\omega(G)$. By \rref{thm:clique-lower-bound} and the simple algebraic fact that $\forall a \in \mathbb{N}$, $a \leq 3 \lfloor \frac{a}{2} \rfloor$ (which can be proven simply by noting that if $a$ is even $a \leq  3 \frac{a}{2}$ and if $a$ is odd $a \leq  3 \frac{a-1}{2}$., $\omega(G) \leq 3 m(D)$, Algorithm \ref{alg:phase1} uses at most $3 m(D)$ single-node interventions. Next, by Lemma \ref{lemma:vis-characterization} and Lemma \ref{lemma:phase2}, and the fact that $\forall a \in \mathbb{N}$, $a - 1 \leq 2 \lfloor {a \over 2} \rfloor$, the second phase of Algorithm \ref{alg:dct-policy} uses at most $2 m(D)$ single-interventions.
\end{proof}

\section{Additional Experimental Results}\label{app:additional-experiments}
\subsection{Scalability of \texttt{OptSingle}}\label{app:optsingle-scalability}
We use the same graph generation procedure as outlined in \rref{section:experiments}. We compare \texttt{OptSingle}, \texttt{Coloring}, \texttt{DCT}, and \texttt{ND-Random} on graphs of up to 25 nodes in \rref{fig:medium-results}. We observe that at 25 nodes, \texttt{OptSingle} already takes more than 2 orders of magnitude longer than either the \texttt{Coloring} or \texttt{DCT} policies to select its interventions, while achieving comparable performance in terms of average competitive ratio.

\begin{figure*}[t!]
    \begin{subfigure}[b]{0.48\textwidth}
         \includegraphics[width=\textwidth]{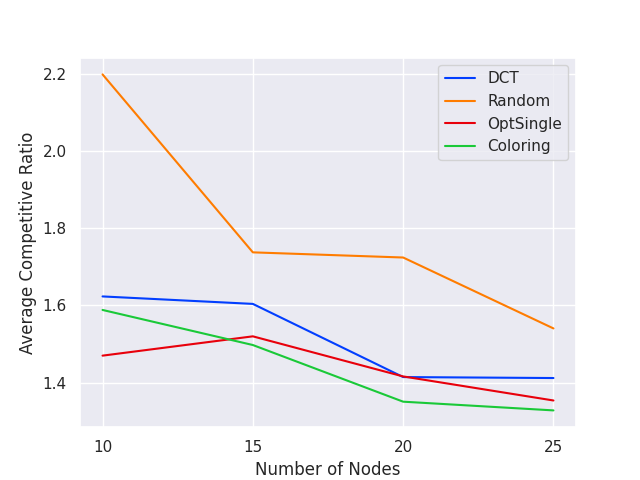}
         \caption{Average ic-ratio}
         \label{fig:avg-regret-medium}
     \end{subfigure}
     ~
     \begin{subfigure}[b]{0.48\textwidth}
         \centering
         \includegraphics[width=\textwidth]{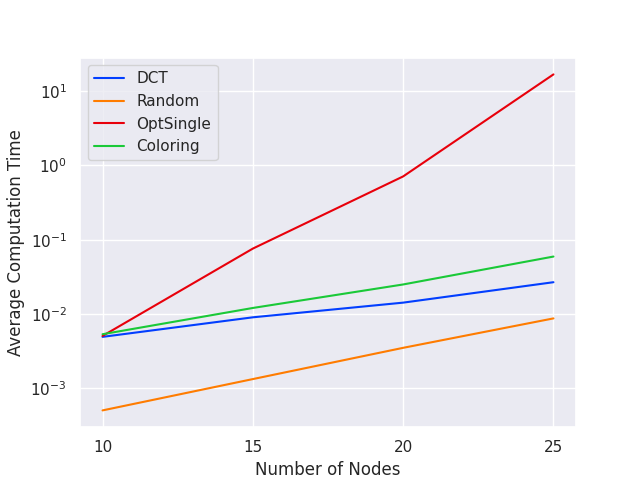}
         \caption{Average Computation Time}
         \label{fig:time-medium}
     \end{subfigure}
     \caption{Comparison (over 100 random synthetic DAGs)}
     \label{fig:medium-results}
\end{figure*}

\subsection{Computation time for large tree-like graphs}\label{app:computation-time-large}

In this section, we report the results on average computation time associated with \rref{fig:avg-regret-large} from \rref{section:experiments}. We find similar scaling for our \texttt{DCT} policy and the \texttt{Coloring} policy, both taking about 5-10 seconds for graphs of up to 500 nodes, as seen in \rref{fig:time-large}.

\subsection{Comparison on large dense graphs}

In this section, we generate dense graphs via the same Erd{\"o}s-R{\'e}nyi-based procedure described in \rref{section:experiments}.
We show in \rref{fig:large-dense-results} that the \DCT policy is more scalable to dense graphs than the \texttt{Coloring} policy, but that our performance becomes slightly \textit{worse} than even \texttt{ND-Random}. 
Since the size of the MVIS is already large for such dense graphs, this suggests that the two-phase nature of the \DCT policy may be too restrictive for such a setting.
Further analysis of the graphs on which different policies do well is left to future work.

\begin{figure*}[t!]
    \centering
    \begin{subfigure}[b]{0.48\textwidth}
         \includegraphics[width=\textwidth]{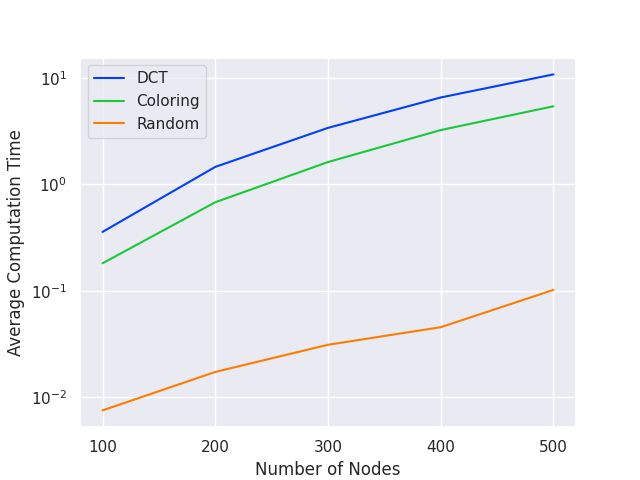}
         \caption{Average Computation Time}
         \label{fig:time-large}
     \end{subfigure}
\end{figure*}

\begin{figure*}[t!]
    \begin{subfigure}[b]{0.48\textwidth}
         \includegraphics[width=\textwidth]{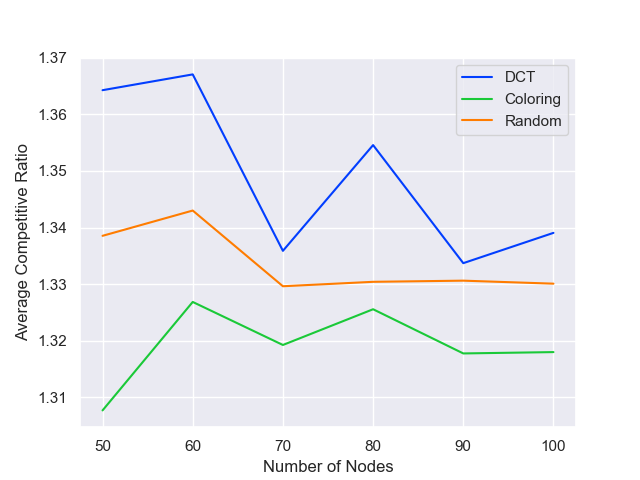}
         \caption{Average ic-ratio}
         \label{fig:avg-regret-large-dense}
     \end{subfigure}
     ~
     \begin{subfigure}[b]{0.48\textwidth}
         \centering
         \includegraphics[width=\textwidth]{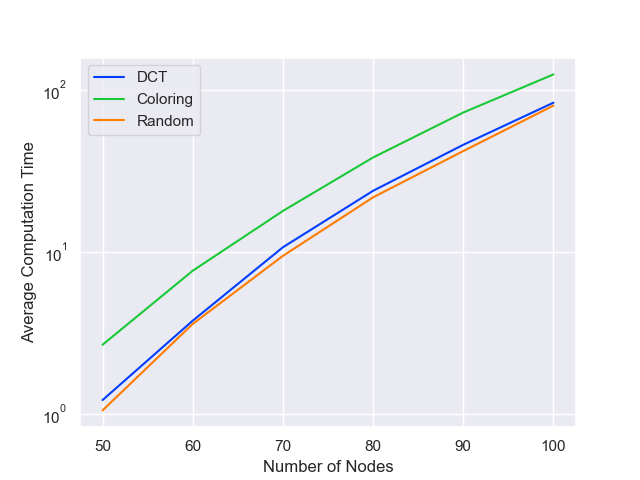}
         \caption{Average Computation Time}
         \label{fig:time-large-dense}
     \end{subfigure}
     \caption{Comparison (over 100 random synthetic DAGs)}
     \label{fig:large-dense-results}
\end{figure*}

\end{document}